\documentclass[journal,12pt,draftclsnofoot,onecolumn]{IEEEtran}
\usepackage{mathtools}
\usepackage{cite}
\ifCLASSINFOpdf
\fi
\usepackage{amsmath,amsfonts,amsthm} 
\usepackage{amssymb}
\usepackage{bbm}
\usepackage{float}
\usepackage{graphicx}
\usepackage{algorithm}
\usepackage{algpseudocode}
\usepackage{subcaption}
\usepackage{siunitx}
\usepackage{tikz}
\usetikzlibrary{arrows,automata}
\usepackage{multicol}
\usepackage{comment}

\usepackage{pgfplots}
\usepackage{epstopdf}
\usepackage{multirow}
\usepackage[font=scriptsize,skip=5pt]{caption}
\usepackage{soul}
\usepackage{todonotes}
\usepackage{color, colortbl}
\usepackage{array}
\usetikzlibrary{patterns}
\usepackage{balance}

\usepackage{enumitem}
\usepackage{footmisc}

\usetikzlibrary{shapes.geometric}

\newcommand{\argmax}{\arg\!\max}

\graphicspath{{images/}}
\newtheorem{theorem}{Theorem}

\newtheorem{lemma}[theorem]{Lemma}
\newtheorem{proposition}[theorem]{Proposition}

\makeatletter
	{\end{proof}%
	\renewcommand{\qedsymbol}{\qedsymbol}}
\makeatother

\allowdisplaybreaks

\DeclareMathAlphabet\mathbfcal{OMS}{cmsy}{b}{n}

\definecolor{Gray}{gray}{0.9}
\definecolor{LightCyan}{rgb}{0.88,1,1}

\algnewcommand{\IIf}[1]{\State\algorithmicif\ #1\ \algorithmicthen}
\algnewcommand{\IElse}[2]{\State\algorithmicelse\ #2\ }
\algnewcommand{\EndIIf}{\unskip\ \algorithmicend\ \algorithmicif}

\IEEEoverridecommandlockouts
\begin{document}
	\setlength{\parskip}{0em}
	
	\title{Throughput Maximization with an Average  Age of Information Constraint in Fading Channels
	}
	\author
	{
	\IEEEauthorblockN
	{   Rajshekhar Vishweshwar Bhat,~\IEEEmembership{Member,~IEEE,}
	    and Rahul Vaze,~\IEEEmembership{Member,~IEEE}\\
	    and Mehul Motani,~\IEEEmembership{Fellow,~IEEE}
	}
}
	\maketitle
\begin{abstract}
In the emerging fifth generation (5G) technology, communication nodes are expected to support two crucial classes of information traffic, namely, the enhanced mobile broadband (eMBB) traffic with high data rate requirements,  and ultra-reliable low-latency communications (URLLC) traffic with strict requirements on latency and reliability. The URLLC traffic, which is usually analyzed by a metric called the age of information (AoI), is assigned the first  priority over the resources at a node. Motivated by this, we consider long-term average throughput maximization problems subject to average AoI and power constraints in a single user fading channel, when (i) perfect and (ii) no channel state information at the transmitter (CSIT) is available.  
We propose simple \emph{age-independent} stationary randomized policies (AI-SRP),  which allocate powers at the transmitter based only on the channel state and/or distribution information, without any knowledge of the AoI. We show that the optimal throughputs achieved by the AI-SRPs for scenarios (i) and (ii) are at least equal to the half of the respective optimal long-term average throughputs, independent of all the parameters of the problem, and that they are within additive gaps, expressed in terms of the optimal dual variable corresponding to their average AoI constraints, from the respective optimal long-term average throughputs.    
 \end{abstract}
	\IEEEpeerreviewmaketitle
\section{Introduction}
In the   Internet of Things (IoT) applications supported by the emerging fifth generation (5G) technology, multiple information streams may need to be transmitted from a source to a destination, in which certain information streams may be more critical than others in the sense that they may require \emph{timely} delivery with a certain quality requirement \cite{5G-Survey1,5G-Survey2,Krishna_IITM,Bennis}.  For instance, real-time URLLC  applications, such as the autonomous driving, remote surgery, augmented reality, haptic communication and disaster management,  require timely information delivery \cite{5G-Survey1}. 
The timeliness is measured by the age of information (AoI), the time elapsed since the generation of the last successful update \cite{Krishna_IITM,Bennis}, and the \emph{quality} of information is quantified by the number of bits delivered for an update \cite{AQI}. 
In these applications,  timely but \emph{low-quality} updates are not useful for subsequent actions. Hence, it is sensible to deem an update to be successful only if the number of bits delivered per update is above a threshold.  
Once a fraction of the available resources has been allocated to satisfy the above    requirements,  expressed in terms of an average AoI, the remaining resources at a node can be used for transmission of other information streams, such as eMBB traffic, in a best-effort manner \cite{5G-Survey1,5G-Survey2,Krishna_IITM}.  
Hence, in this work, we aim to maximize a long-term average throughput subject to a long-term average AoI. The resource constraint we account for at the node is a long-term average power constraint.

The optimization of AoI and throughputs, independently, over wireless channels has been widely studied in the recent and past years \cite{AoI-monograph,Tse}. However,  
 AoI minimization and throughput maximization  problems subject to throughput and age constraints, respectively, have not been widely studied. 
The authors in \cite{Uplink-Modiano, Uplink-Modiano1} consider a problem of AoI minimization at an access point in an uplink wireless communication channel subject to   throughput constraints on information transmitted from source nodes.   The authors prove that a simple  age-independent stationary randomized policy (AI-SRP) has at most twice the average AoI of the optimal policy.  The authors in \cite{Globecom-18} study a multiple access channel with two objectives: minimization of an average AoI of
 \emph{age-critical} primary nodes subject to an aggregate  throughput
requirement in secondary nodes and maximization of an aggregate secondary  user throughput subject to an average AoI constraint. 
Moreover, the authors in \cite{Krishna_IITM} consider the problem of minimizing peak AoI subject to a throughput constraint  over a downlink erasure channel. 
 The authors in \cite{Modian-MobiCom} consider a wireless network with a  base station serving packets from multiple streams to multiple destinations over a wireless channel. They derive a lower bound on an average AoI using the queuing theoretic framework. 
\cite{HARQ} proposes scheduling policies for the transmission of status updates over an error-prone communication channel to minimize a long-term average AoI under a constraint on the average number of transmissions at the source.  
The authors in \cite{MarkovChannels} consider a wireless fading channel modeled as a two-state Markov chain, where channel alternates between a \emph{good} and a \emph{bad} state with certain transition probabilities, and derive a closed-form expression for an average AoI.

Unlike in the above works, where communication occurs over simple two-state channels,  we consider a  general wireless fading channel with a random channel power gain. 
Moreover, 
our goal is to obtain transmit powers in every slot to optimize the considered problem, unlike in \cite{Uplink-Modiano, Uplink-Modiano1, Globecom-18, Modian-MobiCom, Krishna_IITM,HARQ}, where, the only decision to be made is whether to transmit or not in a given slot. 
To the best of our knowledge, the only work that considers a power adaptation to minimize an average AoI is \cite{Age-PowerControl}, where   transmit power is varied based on the number of negative acknowledgments received by the transmitter,  communicating over a fading channel without CSIT. 
Unlike in \cite{Age-PowerControl}, our goal is to adapt transmit powers based on  the channel power gain realization and/or distribution.

The main contributions of the work are as follows. 
\begin{itemize}
	\item We propose AI-SRPs that maximize long-term average throughputs subject to average AoI and power constraints, based on the knowledge of the current channel state and distribution information, when perfect CSIT is available, and only the channel distribution information, when no CSIT is available, in a single-user fading channel. 
	
	\item  For the proposed AI-SRPs, when perfect CSIT is available, we provide a closed-form expression for the optimal transmit power. When no CSIT is available, we adopt a broadcast strategy where codewords are designed for each possible channel realizations and superimposed before transmissions, and depending on the channel realization, the receiver can reliably recover a fraction of the codewords. 

	\item We show that the maximum long-term average throughputs achieved in the proposed AI-SRPs under perfect and no CSIT cases are at least half of the respective optimal long-term average throughputs, independent of all the parameters of the problems, and that they are within additive gaps  from the  respective optimal long-term average throughputs. We express the additive gaps in terms of the optimal dual variables  corresponding to the respective AoI constraints. 
\end{itemize}

\emph{In summary, we find that   age-independent policies are near-optimal for maximization of long-term average throughputs subject to average AoI and power constraints in fading channels.}  

The remainder of the paper is organized as follows. We present the system model, problem formulation and an upper-bound to the problem in Section~\ref{sec:system-model}. 
 We propose    AI-SRPs and bound their performance for the perfect and no CSIT cases in Section~\ref{sec: perfect CSIT} and Section~\ref{sec: NO CSIT}, respectively. We present numerical results in Section~\ref{sec: numerical results} and conclude in Section \ref{sec:conclusions}.

\section{System Model and Problem Formulation}\label{sec:system-model}
In this section, we first describe the channel, rate and AoI models considered in the work. We then present the long-term average AoI and power constraints followed by the long-term average throughput maximization problem subject to these constraints. We finally present an optimization problem whose optimal objective value is an upper-bound to that of the original problem, based on a technique proposed in \cite{Uplink-Modiano}. We will exploit this \emph{upper-bound} to obtain bounds on the AI-SRPs that we propose in the subsequent sections. 
\subsection{Channel Model}
We consider a block-fading channel in which the channel power gain realization remains constant over a fixed duration of time, referred to as the channel coherence block length, and changes across blocks independently. 
Let the random channel power gain in block $k$ be denoted by $H(k)$. We assume that  $\{H(1),H(2),\ldots\}$ is a sequence of independent and identically distributed (i.i.d.) random variables, distributed as $H$ over alphabets in  $\mathcal{H}$, a compact subset of $\mathbb{R}^+$, where $\mathbb{R}^+$ is the set of positive non-zero real numbers.   
 We consider two cases: (i) when the CSIT is perfectly available  and (ii) when no CSIT is available, at the start of a channel coherence block. For ease of presentation, we consider $H$  to be a continuous random variable  with the probability density function, $f_h$, where $\int_{h\in \mathcal{H}}f_hdh=1$ for (i), and a discrete random variable with $N$ non-zero discrete realizations such that $H=h_i$ with probability $p_i>0$ for all $i\in \{1,\ldots,N\}$, where $h_i$'s are ordered such that  $0<h_1<h_2<\ldots<h_N$ and $\Sigma_{i=1}^Np_i=1$ for (ii). The results presented in the paper for (i) and (ii) can be easily adapted for  discrete and continuous channel power gain distributions, respectively.   
 Without loss of generality, we assume that the length of each coherence block is  unity. 

\subsection{Rate Model}
In the emerging latency and age-critical applications, it is important to adopt packets with short lengths to minimize overheads, such as transmission and processing delays \cite{5G-Survey1,5G-Survey2}. Motivated by this, we consider the case when the status update packets (codewords) span a single channel coherence block.
Let $P(k)$ and $R_p(k)$ respectively be the transmit  power and transmit rate in frame $k$. In this case, depending on the channel realization, only a fraction of the transmitted bits can be successfully decoded. Let $R_s(k)\leq R_p(k)$ be the number of bits that are reliably decoded in frame $k$, which can be possibly random. 
Then, the long-term average number of bits successfully delivered, referred to as the long-term average throughput, is given by
\begin{align}\label{eq:avg-throughput}
\bar{R}=\lim_{K\rightarrow \infty}\frac{1}{K}\sum_{k=1}^{K}\mathbb{E}[R_s(k)],
\end{align}
where the expectation is with respect to the channel power gain process $\{H(1),H(2)\ldots\}$.

\subsection{AoI Model}
Let the instantaneous AoI in block $k$ be denoted by $a(k)$. 
We deem an  update to be successful in frame $k$ if the number of bits reliably decoded in frame $k$, $R_s(k)$, is strictly greater than  $R_0$.
Whenever an update is successful, the age of information (AoI) is reset to $1$, otherwise, the AoI increases by $1$.   That is, the AoI evolves as follows: 
\begin{align}\label{eq:age_evolution}
a(k+1)= 
\begin{cases}
1& R_s(k)> R_0,\\
a(k)+1             & \text{otherwise}.  
\end{cases}
\end{align}
Without loss of generality, we assume $a(0)= 0$. Then, the long-term average AoI is given by 
\begin{align}\label{eq:Avg-AoI}
\bar{a}\triangleq \lim_{K\rightarrow \infty}\frac{1}{K}\sum_{k=1}^{K}\mathbb{E}[a(k)], 
\end{align}
where the expectation  is with respect to the channel power gain process $\{H(1),H(2)\ldots\}$.

\subsection{Constraints}
As discussed in the introduction, in various emerging applications, the  timeliness of information, measure by the AoI, is an important requirement. In order to  incorporate this, we impose the following average AoI constraint:  
$\bar{a}\leq \alpha$.
 We also impose the following average power constraint: 
\begin{align}\label{eq:Avg-Power}
\lim_{K\rightarrow \infty}\frac{1}{K}\sum_{k=1}^{K}\mathbb{E}[P(k)]\leq \bar{P},
\end{align}
where we recall that $P(k)$ is the transmit power in frame $k$. 
The expectation in    \eqref{eq:Avg-Power} is with respect to the channel power gain process $\{H(1),H(2)\ldots\}$. 

\subsection{Average Throughput Maximization}
With the above setting, our goal is to maximize the long-term average throughput subject to the long-term average AoI and long-term average power constraints, which can be accomplished by solving the following optimization problem:  

\begin{subequations}\label{eq:main-opt-problem}
	\begin{align}
	R^*=\underset{P(1), P(2), \ldots}{\text{max}} &\;\; \lim_{K\rightarrow \infty}\frac{1}{K}\sum_{k=1}^{K}\mathbb{E}[R_s(k)],  \;\;\label{eq:throughput}\\
	\text{subject to}&\;\; \lim_{K\rightarrow \infty}\frac{1}{K}\sum_{k=1}^{K}\mathbb{E}[a(k)]\leq \alpha, \label{eq:avg-AoI}\\
	&\;\; \lim_{K\rightarrow \infty}\frac{1}{K}\sum_{k=1}^{K}\mathbb{E}[P(k)]\leq \bar{P}, \label{eq:avgP}
	\end{align}
\end{subequations}
where \eqref{eq:avg-AoI} and \eqref{eq:avgP} are the average AoI and power constraints, respectively.

The optimization problem in \eqref{eq:main-opt-problem} has the structure of a constrained Markov decision process \cite{VSharma,CMDP}, with the state space given by the set $\mathcal{H}\cup \mathcal{A}$, where $\mathcal{A}$ is the set of all possible realizations of age, and the action space is given by $\mathbb{R}^+\cup \{0\}$. 
In slot $k\in \{1,2,\ldots\}$, the policy $P(k)$ is a mapping (possibly random) from $(h(1), a(1),  P(1),  h(2), a(2), P(2),\ldots,h(k-1), a(k-1), P(k-1), h(k), a(k))$ to $\mathbb{R}^+\cup \{0\}$.  
Let $\Pi$ be the space of all such policies. A \emph{stationary policy} is of the form $\pi=(\phi,\phi,\ldots)$, where $\phi$ is a mapping (possibly random) from $\mathcal{H}\cup \mathcal{A}$ to $\mathbb{R}^+\cup \{0\}$.  
\emph{In this work, we are interested in  an {AI-SRP} with a bounded performance gap from the optimal $R^*$ in \eqref{eq:main-opt-problem}. An AI-SRP is of the form $\pi'=(\psi,\psi,\ldots)$, where $\psi$ is a mapping (possibly random) from $\mathcal{H}$ to $\mathbb{R}^+\cup \{0\}$.}  

Before we proceed to obtain AI-SRPs for the perfect and no CSIT cases in the subsequent sections, in the following, we present an upper-bound to $R^*$, which we will eventually  exploit to bound the performance of AI-SRPs.

\subsection{An Upper Bound on \eqref{eq:main-opt-problem}}\label{sec:UB} 
Define the following indicator variable:  
\begin{align}\label{eq:u}
u(k)\triangleq  
\begin{cases}
1& R_s(k)> R_0,\\
0             & \text{otherwise},   
\end{cases}
\end{align}
and consider the following optimization problem:  
\begin{subequations}\label{eq:upper-bound-problem}
	\begin{align}
	U=\underset{P(1), P(2), \ldots}{\text{max}} &\;\; \lim_{K\rightarrow \infty}\frac{1}{K}\sum_{k=1}^{K}\mathbb{E}[R_s(k)],  \;\;\\
	\text{subject to}&\;\; \lim_{K\rightarrow \infty}\frac{1}{K}\sum_{k=1}^{K}\mathbb{E}[u(k)]\geq \frac{1}{2\alpha -1},\label{eq:NofUpdates} \\
	&\;\; \lim_{K\rightarrow \infty}\frac{1}{K}\sum_{k=1}^{K}\mathbb{E}[P(k)]\leq \bar{P}.\label{eq:avgPub} 
	\end{align}
\end{subequations}
We now have the following result. 
\begin{proposition}\label{prop:UB}
	The optimal value of \eqref{eq:upper-bound-problem}, $U$, is an upper-bound to the optimal objective value of \eqref{eq:main-opt-problem}. 
\end{proposition}
\begin{proof}
 For obtaining the result, we first note the fact that 
\begin{align}\nonumber
\lim_{K\rightarrow \infty}\frac{1}{K}\sum_{k=1}^{K}\mathbb{E}[a(k)] \geq \frac{1}{2}\frac{1}{\lim_{K\rightarrow \infty}\frac{1}{K}\sum_{k=1}^{K}\mathbb{E}\left[u(k)\right]}+1,
\end{align}
which can be proved along the lines in the proof of Theorem 1 in \cite{Uplink-Modiano}. 
Now, applying the constraint on the long-term average AoI in \eqref{eq:avg-AoI}, and re-arranging the terms, we get \eqref{eq:NofUpdates}. By doing so, since we are constraining a lower bound on the long-term average AoI, the optimal solution with \eqref{eq:NofUpdates} is an upper-bound to that of  \eqref{eq:main-opt-problem}.  See  Appendix A   for details.   
\end{proof}
From the definition of $u(k)$ in \eqref{eq:u}, the left-hand side of \eqref{eq:NofUpdates} physically signifies the long-term average number of successful updates.

\section{Perfect CSIT Case}\label{sec: perfect CSIT}
Since perfect CSIT is available, we can always transmit
at the maximum rate of reliable communication, i.e., we can
transmit with rate $R_p(k) = R_s(k)$. When the
channel realization in frame $k$ is $h(k)$ and transmit power is
$P(k)$, we assume $R_s(k)=r(h(k)P(k))$, where $r(x)$ is the maximum rate of reliable communication when the received signal-to-noise ratio is $x$, and $r(\cdot)$ is an invertible, concave, strictly increasing function defined over $[0,\infty)$ such that $r(0)=0$. 
 Hence, when perfect CSIT is available, \eqref{eq:main-opt-problem} can be re-written as
\begin{subequations}\label{eq:main-opt-problemCSIT}
	\begin{align}
	R^*_{\rm CSIT}=\underset{P(1), P(2), \ldots}{\text{max}} &\;\; \lim_{K\rightarrow \infty}\frac{1}{K}\sum_{k=1}^{K}\mathbb{E}[r(h(k)P(k))],\label{eq:thr}  \\
		\text{subject to}&\;\; \eqref{eq:avg-AoI}, \eqref{eq:avgP}, 
	\end{align}
\end{subequations}
and the upper-bound problem in \eqref{eq:upper-bound-problem} can be re-written as
\begin{subequations}\label{eq:upper-bound-problemCSIT}
	\begin{align}
	U_{\rm CSIT}=\underset{P(1), P(2), \ldots}{\text{max}} &\;\; \lim_{K\rightarrow \infty}\frac{1}{K}\sum_{k=1}^{K}\mathbb{E}[r(h(k)P(k))],  \;\;\\
	\text{subject to}&\;\; \eqref{eq:NofUpdates},\eqref{eq:avgPub}.  
	\end{align}
\end{subequations}
In the remaining part of the section, we first propose an AI-SRP. We  then transform the upper-bound problem in \eqref{eq:upper-bound-problemCSIT} to an equivalent problem, which is identical to the proposed AI-SRP with the average AoI constraint of $\alpha$ relaxed to $2\alpha -1$. We finally note that the optimal throughput of the proposed AI-SRP is concave in $\alpha$ and use this fact to obtain additive and multiplicative bounds on its performance.

\subsection{AI-SRP}\label{sec:AI-SRP CSIT}
In this policy, when the channel power gain is $h$, we transmit with power
\begin{align}\label{eq:AISRP-policy}
P_h=  
\begin{cases}
P_h^1 & \text{with probability}\; \mu_h,\\
P_h^2            & \text{otherwise},   
\end{cases}
\end{align}
where $P_h^1> {r^{-1}(R_0)}/{h}$ and $P_h^2\leq {r^{-1}(R_0)}/{h}$. 
Clearly, when $P_h=P_h^1$, the update will be successful always, as $R_s>R_0$ and the AoI will be reset to $1$. Hence, the probability of success, $\mathbb{P}(R_s> R_0)=\mathbb{E}[\mu_h]$. Since the channel power gains are i.i.d. across slots, the \emph{inter-success interval}, the number of blocks between two consecutive successful updates,  is geometrically distributed random variables with mean $1/\mathbb{E}[\mu_h]$ and we can evaluate the long-term average AoI  in \eqref{eq:Avg-AoI}  to be  $1/\mathbb{E}[\mu_h]$ \cite{Uplink-Modiano}. Now, the average AoI constraint in \eqref{eq:avg-AoI} can be rewritten as   $\mathbb{E}[\mu_h]\geq 1/\alpha$. Further, the long-term average throughput in \eqref{eq:thr} and the average power constraint in \eqref{eq:avgP}   can be re-written as 
$\int_{h\in \mathcal{H}} \left(\mu_hr(hP_h^1)+(1-\mu_h)r(hP_h^2)\right)f_hdh$ and  $\int_{h\in \mathcal{H}}   \left(\mu_hP_h^1+(1-\mu_h)P^2_h\right)f_hdh \leq \bar{P}$, respectively. 
We can now rewrite \eqref{eq:main-opt-problemCSIT} as follows: 
\begin{subequations}\label{eq:stat_rand_CSIT} 
	\begin{align}
	\underset{P_h^1,P_h^2,\mu_h}{\text{max}} &\int_{h\in \mathcal{H}} \left(\mu_hr(hP_h^1)+(1-\mu_h)r(hP_h^2)\right)f_hdh,  \;\;&\label{eq:objSRP}\\
	\text{subject to}&\; \int_{h\in \mathcal{H}}\mu_h f_hdh\geq\frac{1}{\alpha},~ 0\leq \mu_h	\leq 1, \label{eq:AOI} \\
	&  \int_{h\in \mathcal{H}}   \left(\mu_hP_h^1+(1-\mu_h)P^2_h\right)f_hdh \leq \bar{P}, \label{eq:Pc}\\
	& P_1^h\geq  \frac{r^{-1}(R_0+\epsilon)}{h},~ 0\leq P_2^h\leq   \frac{r^{-1}(R_0)}{h},
	\end{align}
\end{subequations}
for an arbitrarily small $\epsilon>0$. 
The above problem is non-convex, i.e., it is not jointly convex in $P_h^1,P_h^2$ and $\mu_h$. To transform it into a jointly convex problem, define $v^1_h\triangleq \mu_hP^1_h$ and $v^2_h\triangleq (1-\mu_h)P^2_h$. We can now re-write \eqref{eq:stat_rand_CSIT} as the following jointly convex optimization problem:   
\begin{subequations}\label{eq:stat_rand_CSIT_convex}  
	\begin{align}
	R_{\rm CSIT}^{\rm AI\_SRP}\left(\alpha\right)=
	\underset{\mu_h, v^1_h, v^2_h\geq 0}{\text{max}} &\;\;\int_{h\in \mathcal{H}} \left(\mu_hr\left(\frac{hv_h^1}{\mu_h}\right)+(1-\mu_h)r\left(\frac{hv_h^2}{1-\mu_h}\right)\right)f_hdh,  \;\;\label{eq:obj}\\
	\text{subject to}&\; \int_{h\in \mathcal{H}}\mu_h f_hdh\geq \frac{1}{\alpha},~ 0\leq \mu_h	\leq 1,\label{eq:outage} \\
	&\;  \int_{h\in \mathcal{H}}   \left(v_h^1+v^2_h\right)f_hdh \leq \bar{P},\label{eq:avp} \\
	&\; v_1^h\geq \mu_h\frac{r^{-1}(R_0+\epsilon)}{h},\\
	&\; v_2^h\leq  (1-\mu_h) \frac{r^{-1}(R_0+\epsilon)}{h},~ v^2_h\geq 0, 
	\end{align}
\end{subequations} 
for an arbitrarily small $\epsilon>0$.


\subsection{An Analytical Optimal Solution for the AI-SRP}
Let $\mu_h^*, v^{1*}_h, v^{2*}_h$ respectively be the optimal $\mu_h, v^1_h, v^2_h$ that solve  \eqref{eq:stat_rand_CSIT_convex}. Then the optimal solution to \eqref{eq:stat_rand_CSIT} is given by $P^{1*}_h= v^{1*}_h/\mu_h^*$ and $P^{2*}_h= v^{2*}_h/\mu_h^*$. 
In the following theorem, we now present the optimal $u_h^*$, $P^{1*}_h$ and $P^{2*}_h$. 
\begin{theorem}\label{prop:CI-WF}
	For some $\lambda\geq 0$, let  
	\begin{align*}
		P_{\lambda,h} &\triangleq 	\min\{P:hr'(hP)=\lambda, P\geq 0\},\;\; \forall\, h\in \mathcal{H},\\
		h_{\alpha}&\triangleq \max\left\{h:\int_{h}^{\infty} f_hdh =\frac{1}{\alpha}\right\},  \\
		h_{\lambda}&\triangleq \min\left\{h:h=\frac{r^{-1}(R_0+\epsilon)}{P_{\lambda,h}}\right\},
	\end{align*}
	where $r'(x)$ is the first derivative of $r(x)$ with respect to $x$. 
	Then, we have
	\begin{align}\label{eq:optPower}
		\mu_h^*&=  
		\begin{cases}
			1, & \forall\, h\geq h_{\alpha},\\
			0,             &  \text{otherwise}, 
		\end{cases}, \nonumber\\
		P_h^{1^*}&=  
		\begin{cases}
			\frac{r^{-1}(R_0+\epsilon)}{h}, & \forall\, h_{\alpha}\leq  h<  	h_{\lambda},\\
			P_{\lambda,h},             &  \forall\, h\geq h_{\lambda}, 
		\end{cases},
		\;\;\; 		P_h^{2^*}= P_{\lambda,h},\; \forall \,  h\leq h_{\alpha}, 
	\end{align}	
	for an arbitrarily small $\epsilon>0$, where  $\lambda$ is chosen to satisfy the average power constraint in \eqref{eq:Pc}, or equivalently,  \eqref{eq:avp}.   
\end{theorem}

\begin{proof}
	The result can be obtained by first finding $\mu_h^*, v^{1*}_h, v^{2*}_h$  of  \eqref{eq:stat_rand_CSIT_convex} and then assigning $P^{1*}_h= v^{1*}_h/\mu_h^*$ and $P^{2*}_h= v^{2*}_h/\mu_h^*$.  
	The average AoI constraint in \eqref{eq:outage} can be treated as an outage probability constraint, when an outage event is said to occur if $R_s(k)\leq R_0$.  With this observation, if we let $r(x)=\log(1+x)$, the optimization in \eqref{eq:stat_rand_CSIT_convex} is in fact the throughput maximization problem subject to an outage probability and average power constraints studied in \cite{ServiceOutage}, whose optimal solution is given in Theorem 1 of  \cite{ServiceOutage} using the Karush-Kuhn-Tucker (KKT) conditions. 
	In other words, Theorem 1 of  \cite{ServiceOutage} provides the optimal solution to \eqref{eq:stat_rand_CSIT_convex}  when $r(x)=\log(1+x)$. 
	Following the proof of Theorem 1 of  \cite{ServiceOutage},  we can obtain the solution of \eqref{eq:stat_rand_CSIT_convex} for any invertible concave increasing function $r(x)$ with $r(0)=0$, as stated in \eqref{eq:optPower}. 
\end{proof}

We now make the following comments on the above result. 
Note that the proposed AI-SRP in \eqref{eq:optPower} is \emph{achievable} as it satisfies both the average power and AoI constraints, by design. 
Further, when $r(x)=\log(1+x)$, as mentioned in the proof of Theorem~\ref{prop:CI-WF},    \eqref{eq:optPower} specializes to the following power allocation policy obtained in Theorem 1 of  \cite{ServiceOutage}: 
\begin{align}\label{eq:optPower1}
	P_h^*=  
	\begin{cases}
		\frac{\exp(R_0+\epsilon)-1}{h}, & \forall\; h_{\alpha}\leq  h<  \lambda\exp(R_0+\epsilon),\\
		\left(\frac{1}{\lambda}-\frac{1}{h}\right)^+,            &  \text{otherwise}.   
	\end{cases}
\end{align}
The optimal power allocation given in \eqref{eq:optPower1} can be interpreted as a combination of the   well-known channel inversion and water-filling algorithms \cite{ServiceOutage,Tse}. To see this, note that, for $h\in [h_{\alpha}, \lambda\exp(R_0+\epsilon))$, the optimal transmit power is inversely proportional to the channel power gain, similar to the channel inversion algorithm. Otherwise, the optimal transmit power is identical to the water-filling algorithm with water level $1/\lambda$.

\subsection{Multiplicative and Additive Bounds on the Performance of AI-SRP}
In this subsection, we obtain additive and multiplicative bounds on the AI-SRP proposed in \eqref{eq:optPower}. We begin by  expressing the optimal objective value of  the   upper-bound problem, $U_{\rm CSIT}$ in  \eqref{eq:upper-bound-problemCSIT} in terms of $R_{\rm CSIT}^{\rm AI\_SRP}(\cdot)$ in \eqref{eq:stat_rand_CSIT_convex}.

\begin{theorem}\label{thm:UB-CSIT}
	The optimal objective value of the upper-bound problem in  \eqref{eq:upper-bound-problemCSIT},
	\begin{align*}
	U_{\rm CSIT}=R_{\rm CSIT}^{\rm AI\_SRP}(2\alpha-1), 
	\end{align*}
	where $R_{\rm CSIT}^{\rm AI\_SRP}(\cdot)$ is given by  \eqref{eq:stat_rand_CSIT_convex}. 
\end{theorem}
\begin{proof}
	To prove the above theorem, we first note that \eqref{eq:upper-bound-problemCSIT} has the form of a constrained Markov decision process (MDP). We then transform it to an unconstrained MDP using the Lagrangian method and note that it will have a stationary randomized policy (SRP) as the optimal policy. We then show that the most general SRP one can employ in this case is similar to the AI-SRP proposed in the previous subsection. 
	See Appendix B for the details.  	
\end{proof}

We now have the following main result of this section, where we obtain multiplicative and additive bounds on the performance of the  AI-SRP in \eqref{eq:optPower}.

\begin{theorem}\label{thm:rp}
	The following bounds hold true for the long-term average throughput achieved by the  AI-SRP  in \eqref{eq:optPower}, obtained as the optimal solution to \eqref{eq:stat_rand_CSIT_convex}. 
	\begin{itemize}
		\item For any $(R_0, H, \alpha,\bar{P})$, we have, $R^*_{\rm CSIT}\leq 2R_{\rm CSIT}^{\rm AI\_SRP}$. 
		\item $R^*_{\rm CSIT}-R_{\rm CSIT}^{\rm AI\_SRP}\leq  \nu_{\alpha}(\alpha-1)$, where $\nu_{\alpha}$ is the optimal dual variable corresponding to the AoI constraint in \eqref{eq:outage}.
	\end{itemize}
\end{theorem}
\begin{proof}
	Consider \eqref{eq:stat_rand_CSIT_convex} with $\alpha$ being replaced by $x$. Hence, when $x=\alpha$ and $x=2\alpha-1$, we obtain $R_{\rm CSIT}^{\rm AI\_SRP}$ and $U_{\rm CSIT}$ (follows from Theorem~\ref{thm:UB-CSIT}), respectively. 
	From Exercise 5.32 in \cite{Boyd}, since \eqref{eq:stat_rand_CSIT_convex} is an concave maximization problem, the optimal value function, $R_{\rm CSIT}^{\rm AI\_SRP}(\cdot)$ is a concave non-increasing function of $1/x$. Moreover, since $1/x$ is convex in $[1,\infty)$ and because the composition of a convex function and a concave non-increasing function is a  concave function (see Section 3.2 in \cite{Boyd}), we conclude that $R_{\rm CSIT}^{\rm AI\_SRP}(x)$ is concave in $x$. Since, increasing the $x$ is relaxing the average AoI constraint, $R_{\rm CSIT}^{\rm AI\_SRP}(x)$ must be non-decreasing. Now, noting that for any concave function, $g(x),~x\in\mathbb{R}$, the inequality,  $g(y+\delta)-g(y)\leq g(x+\delta)-g(x)$ holds, for all $y\geq x$ and $\delta\geq 0$, and letting, $y=\alpha$, $x=1$ and $\delta=\alpha-1$, we have, $R_{\rm CSIT}^{\rm AI\_SRP}(2\alpha-1)-R_{\rm CSIT}^{\rm AI\_SRP}(\alpha)\leq R_{\rm CSIT}^{\rm AI\_SRP}(\alpha)-R_{\rm CSIT}^{\rm AI\_SRP}(1)$. Since, $R_{\rm CSIT}^{\rm AI\_SRP}(1)\geq 0$, we have  
	\begin{align}\label{eq:R21}
	R_{\rm CSIT}^{\rm AI\_SRP}(2\alpha-1)\leq 2 R_{\rm CSIT}^{\rm AI\_SRP}(\alpha). 
	\end{align}
	Hence, we have
	\begin{align*}
	\frac{R^*_{\rm CSIT}}{  R_{\rm CSIT}^{\rm AI\_SRP}}\leq \frac{U_{\rm CSIT}}{R_{\rm CSIT}^{\rm AI\_SRP}(\alpha)}=\frac{R_{\rm CSIT}^{\rm AI\_SRP}(2\alpha-1)}{R_{\rm CSIT}^{\rm AI\_SRP}(\alpha)}\leq 2,  
	\end{align*}
	where the last inequality follows from \eqref{eq:R21}. 
\end{proof}
We now make the following comments on the above theorem. 
For any parameter tuple of the problem, $(R_0, H, \alpha,\bar{P})$, the first bound says that the maximum long-term average throughput of AI-SRP is at least half of the optimal long-term average throughput. 
Moreover, for a specific parameter tuple, $(R_0, H, \alpha,\bar{P})$, we can obtain a potentially better bound from the  second result. 
For this, we need to find the optimal dual variables, $\nu_{\alpha}$,  corresponding to the average AoI constraint.

\section{No CSIT Case}\label{sec: NO CSIT}
We now consider the case when the realization of the channel power gain, $h(k)$ is not available at the start of slot $k$. In this case, since we do not have knowledge of the instantaneous channel state information, we allocate the transmit powers based only on the knowledge of the distribution of channel power gain, $H$, by adopting the following broadcast strategy. 

\subsubsection{Broadcast Strategy}
As pointed out in the system model section, in this case, we assume $H$ to be a discrete random variable, which can take $N$ non-zero discrete values,  such that $H=h_i$ with probability $p_i>0$ for all $i\in \{1,\ldots,N\}$,  where $h_i$'s are ordered such that  $0<h_1<h_2<\ldots<h_N$, and $\Sigma_{i=1}^Np_i=1$.

In the broadcast strategy, in a fading block, we first design $N$ codebooks   with rates, 
\begin{align}\label{eq:rate_lsc}
R_{i}=r\left(\frac{h_iP_i}{1+h_i\sum_{j=i+1}^{N}P_j}\right),\ \forall\; i=1,\ldots,N,
\end{align}
where $P_i$ is the constant transmit power of codewords in the $i^{\rm th}$ codebook,    $r(x)$ is the maximum rate of reliable communication when the received signal-to-noise ratio is $x$, 
and   $r(\cdot)$ is an invertible, concave, strictly increasing function defined over $[0,\infty)$ such that $r(0)=0$. 
 We then superimpose and transmit $N$ codewords, one from each of the above codebooks. We refer to the transmitted codeword with rate $R_i$ (from the $i^{\rm th}$ codebook), corresponding to the channel power gain of $h_i$,  as the $i^{\rm th}$ layer. 
At the receiver, we adopt the following successive interference cancellation (SIC) decoding:  We first decode the $1^{\rm st}$  layer, treating the symbols from the other higher layers as noise. Then, the decoded symbols from the $1^{\rm st}$ layer are subtracted from the corresponding symbols of the received codeword, following which  the layer $2$ is decoded treating the higher layers as   noise. This process is continued until the last layer has been decoded.



We now obtain the average achievable rate of the above broadcast strategy.  When the channel power gain, $H=h_j$, the received signal power corresponding to  the $i^{\rm th}$ codeword is $h_jP_i$ for $i,\,j=1,\ldots,N$. In this case, in layer $i$, the maximum rate of reliable communication is $C_{i}^j\triangleq r(1+h_jP_i/(1+h_j\sum_{k=i+1}^{N}P_k))$ for $i,\,j=1,\ldots,N$. This implies that when $H=h_j$, we can reliably decode only the layers $1,\ldots,j$,  as $R_i\leq C_i^j$ for $i=1,\ldots,j$. For layers $i=j+1,\ldots,N$, we have, $R_i> C_i^j$ and hence we cannot reliably decode them. 
Hence,  the number of bits that can be reliably decoded when $H = h_j$ is $\sum_{i=1}^{j}R_{i}$ bits. Note that in the above broadcast technique,  the higher the channel power gain, the higher is the  average rate, and this shows that the broadcast strategy facilitates  adaptation of  communication rates to the realized channel state,  without requiring the CSIT.

Since $\text{prob}[H = h_j] = p_j$, the number of bits reliably decoded, $
R_s$, is a random variable with the following distribution: $R_s=\sum_{i=1}^{j}R_{i}$ with probability $p_j$. Hence, the average achievable rate is $\mathbb{E}[R_s]=\sum_{i=1}^{N}q_iR_{i}$, where $q_i\triangleq\sum_{j=i}^{N}p_j$. 
From \eqref{eq:rate_lsc}, we note that the total transmit power, $\sum_{i=1}^{N}P_i$, and the rates, $R_1,\ldots,R_N$, are related as follows: 
\begin{align}\label{eq:sumPower}
\sum_{i=1}^{N}P_i=\sum_{i=1}^{N}h_i^{-1}r^{-1}(R_i)\Pi_{l=1}^{i-1}\left(r^{-1}(R_l)+1\right). 
\end{align}
Hence, given the total transmit power, $P$, in order to maximize the average rate,   $\mathbb{E}[R_s]$, we need to solve the following optimization problem: 
\begin{subequations}\label{eq:LSC} 
	\begin{align}
	\underset{R_1,\ldots,R_N \geq 0}{\text{max}} &\;\; \sum_{i=1}^{N}q_iR_i, \;\;&\\
	\text{subject to}&\;\;\sum_{i=1}^{N}h_i^{-1}r^{-1}(R_i)\Pi_{l=1}^{i-1}\left(r^{-1}(R_l)+1\right)\leq P.\label{eq:uncon-avgP}
	\end{align}
\end{subequations} 
We now comment on the convexity of \eqref{eq:LSC}. Clearly, the objective function is convex. In \eqref{eq:uncon-avgP}, we encounter products of the form $r^{-1}(R_1)r^{-1}(R_2)\ldots r^{-1}(R_i)$ for $i=1,\ldots,N$. Since $r(\cdot)$ is invertible, strictly increasing and concave, its inverse, $r^{-1}(\cdot)$, is convex and non-decreasing. Moreover, $r^{-1}(\cdot)$ is positive. Now, since the product of convex, non-decreasing positive functions is convex,  $r^{-1}(R_1)r^{-1}(R_2)\ldots r^{-1}(R_i)$ is convex \cite{Boyd}, and hence \eqref{eq:uncon-avgP} is convex. This implies that \eqref{eq:LSC} is convex and it can be solved using standard numerical techniques.  However, when $r(x)=\log(1+x)$, based on the KKT conditions, a simple analytical solution to \eqref{eq:LSC} has been obtained in \cite{NoCSIT}, which we present in Appendix C. We will use  this solution to propose an iterative algorithm for the AI-SRP presented later in the section, when $r(x)=\log(1+x)$.

\subsubsection{Reformulation of \eqref{eq:main-opt-problem} and \eqref{eq:upper-bound-problem} Under the Broadcast Strategy}
When no CSIT is available, we apply the broadcast strategy in \eqref{eq:rate_lsc} in every block for transmission. Let $R_i(k)$ and $P_i(k)$, respectively, be the transmit rate and power for layer $i$ in block $k$. Then the expected throughput in block $k$ with the broadcast strategy is given by $\mathbb{E}[R_s(k)] = \sum_{i=1}^{N}q_iR_i(k)$.  
Now, under the broadcast strategy,  \eqref{eq:main-opt-problem} can be re-written as 
\begin{subequations}\label{eq:main-opt-problemBC}
	\begin{align}
	R^*_{\rm NoCSIT}=\underset{P_i(1), P_i(2), \ldots}{\text{max}} &\;\; \lim_{K\rightarrow \infty}\frac{1}{K}\sum_{k=1}^{K}q_i\mathbb{E}[R_i(k)],  \;\;\label{eq:throughputBC}\\
	\text{subject to}&\;\; \eqref{eq:avg-AoI}, \lim_{K\rightarrow \infty}\frac{1}{K}\sum_{k=1}^{K}\mathbb{E}\left[\sum_{i=1}^{N}P_i(k)\right]\leq \bar{P}, \label{eq:avg}
	\end{align}
\end{subequations}
and the upper-bound problem in \eqref{eq:upper-bound-problem} can be re-written as
\begin{subequations}\label{eq:upper-bound-problemBC}
	\begin{align}
	U_{\rm NoCSIT} = \underset{P(1), P(2), \ldots}{\text{max}} &\;\; \lim_{K\rightarrow \infty}\frac{1}{K}\sum_{k=1}^{K}q_i\mathbb{E}[R_i(k)],  \;\;\\
	\text{subject to}&\;\; \eqref{eq:NofUpdates},  \lim_{K\rightarrow \infty}\frac{1}{K}\sum_{k=1}^{K}\mathbb{E}\left[\sum_{i=1}^{N}P_i(k)\right]\leq \bar{P}.\label{eq:avgPubBC} 
	\end{align}
\end{subequations}\\

In the remaining part of the section, we will first propose a simple  AI-SRP under the broadcast strategy. 
We then present an equivalent problem to \eqref{eq:upper-bound-problemBC} to obtain an upper-bound to  \eqref{eq:main-opt-problemBC} under the broadcast strategy. We finally present  multiplicative and additive bounds on the maximum achievable throughput of the proposed AI-SRP.  We also obtain a simple iterative solution to AI-SRP when  $r(x)=\log(1+x)$ in Appendix D, based on the discussion in Appendix C.

\subsection{AI-SRP Under the Broadcast Strategy}\label{eq:AISRP-NoCSIT}
In this policy, we design $N$ rate tuples, each with $N$ elements, represented by $(R_1^j,\ldots,R_N^j)$, such that 
\begin{align*}
R_{i}^j=r\left(\frac{h_iP_i^j}{1+h_i\sum_{k=i+1}^{N}P_k^j}\right),\ \forall\; i=1,\ldots,N,
\end{align*}
where $\sum_{i=1}^{j}R_i^j> R_0$, $\sum_{i=1}^{j-1}R_i^j\leq R_0$, for all $j=1,\ldots,N$ and $P_i^j$ is the transmit power corresponding to the $i^{\rm th}$ entry in the $j^{\rm th}$ rate tuple.  
In any given fading block, we choose the $j^{\text{th}}$ rate tuple, $(R_1^j,\ldots,R_N^j)$, with probability $\mu_j$. We then design $N$ codebooks in which the rate of the $i^{\rm th}$ codebook is $R_i^j$ and select $N$ codewords (layers), one   from each of the codebooks,  superimpose and transmit them. At the receiver, the SIC decoding is adopted.

We now specialize \eqref{eq:main-opt-problemBC} for the above AI-SRP. Conditioned on the event that codewords generated from the $j^{\text{th}}$ rate-tuple have been transmitted, the average achievable throughput is equal to   $\sum_{i=1}^{N}q_iR_i^j$. Hence, the overall average throughput, averaged over all the possible rate-tuples, is given by $\sum_{j=1}^{N}\mu_j\sum_{i=1}^{N}q_iR_i^j=\sum_{j=1}^{N}\sum_{i=1}^{N}\mu_jq_iR_i^j$.  
 Recall that an update is deemed to be successful if the number of bits delivered in the update is greater than   $R_0$. In the above broadcast strategy, when  the $j^{\text{th}}$ rate-tuple is transmitted and   when the channel power gain is $H = h_k$, the maximum achievable throughput is equal to $\sum_{i=1}^kR_i^j$ and, by design, we have, $\sum_{i=1}^lR_i^j> R_0$ for all $l\geq j$.  
 Hence, conditioned on the event that $j^{\text{th}}$ rate-tuple is transmitted, the probability of success,  $\mathbb{P}(\text{success}|j^{\text{th}}~\text{rate-tuple is transmitted})=q_j$,  as $\mathbb{P}(H\geq h_j)=q_j$. 
 Hence, the probability of success is given by  $\sum_{j=1}^{N}q_j\mu_j$ and, along the lines in the perfect CSIT case in the previous section,  the average AoI constraint in \eqref{eq:avg-AoI} can be re-written as  $\sum_{j=1}^{N}q_j\mu_j\geq 1/\alpha$. 
Finally, recalling that the transmit power for layer $i$, corresponding to the $j^{\text{th}}$ rate-tuple is $P_i^j$, the long-term average power is given by $\sum_{j=1}^{N}\mu_j\sum_{i=1}^{N}P_i^j$. From \eqref{eq:sumPower}, it follows that 
$\sum_{i=1}^{N}P_i^j= \sum_{i=1}^{N}h_i^{-1}r^{-1}(R_i^j)\Pi_{l=1}^{i-1}\left(r^{-1}(R_l^j)+1\right)$.  
Now, under the above AI-SRP, \eqref{eq:main-opt-problemBC} can be specialized to the following optimization problem:   
\begin{subequations}\label{eq:BC-Rate-AoI-21}
	\begin{align}
	\underset{R_1^j,\ldots,R_N^j, \mu_j}{\text{max}} &\;\; \sum_{j=1}^{N}\sum_{i=1}^{N}\mu_jq_iR_i^j, \;\;&\label{eq:AI-SRPobj}\\
	\text{subject to}&\;\sum_{j=1}^{N}\mu_j\left(\sum_{i=1}^{N}h_i^{-1}r^{-1}(R_i^j)\Pi_{l=1}^{i-1}\left(r^{-1}(R_l^j)+1\right)\right)\leq \bar{P},\;\;\label{eq:avgPBC} \\
	&\; \sum_{j=1}^{N}q_j\mu_j\geq \frac{1}{\alpha},\; \sum_{j=1}^{N}\mu_j=1,\label{eq:avgAoIBC}\\
		&\; R_i^j,\; \mu_j\geq 0, \;  \sum_{i=1}^{j}R_i^j\geq R_0+\epsilon,\label{eq:AI-SRPBC-d}
	\end{align}
\end{subequations}
for all $j,i=1,\ldots,N$, for an arbitrarily small $\epsilon>0$. The inequalities in  \eqref{eq:avgPBC} and \eqref{eq:avgAoIBC} are the average power and AoI constraints, respectively. 
The above problem is non-convex due to coupling between the variables. Defining $v_i^j\triangleq \mu_jR_i^j$, the above problem can be transformed to the following jointly convex problem:  
\begin{subequations}\label{eq:BC-Rate-AoI-22}
	\begin{align}
	R_{\rm NoCSIT}^{\rm AI\_SRP}(\alpha)=
	\underset{v_1^j,\ldots,v_N^j, \mu_j}{\text{max}} &\;\; \sum_{j=1}^{N}\sum_{i=1}^{N}q_iv_i^j, \;\;&\\
	\text{subject to}&\;\;\sum_{j=1}^{N} \mu_j\left(\sum_{i=1}^{N}h_i^{-1}r^{-1}\left(\frac{v_i^j}{\mu_j}\right)\Pi_{l=1}^{i-1}\left(r^{-1}\left(\frac{v_l^j}{\mu_j}\right)+1\right)\right)\leq \bar{P},\;\;\label{eq:avgPBC1} \\
	&\;\; \sum_{j=1}^{N}q_j\mu_j\geq \frac{1}{\alpha},\; \sum_{j=1}^{N}\mu_j\leq 1,\label{eq:outageBC}\\
		&\;\; v_i^j,\mu_j\geq 0,\;\forall j,i=1,\ldots,N,\\
	&\;\;  \sum_{i=1}^{j}v_i^j\geq \mu_j(R_0+\epsilon),\;\forall j=1,\ldots,N. 
	\end{align}
\end{subequations} 
Since the optimization problem \eqref{eq:BC-Rate-AoI-22} is jointly convex, it can be solved using standard numerical techniques for any  invertible concave increasing $r(\cdot)$, defined over $[0,\infty)$, such that $r(0)=0$. However, when $r(x)=\log(1+x)$, in Appendix D, we obtain a simple iterative algorithm for solving \eqref{eq:BC-Rate-AoI-22}.

\subsection{Upper Bound}
In the following theorem, we express the optimal objective value of  the   upper-bound problem in  \eqref{eq:upper-bound-problemBC} in terms of the optimal objective value, $R_{\rm NoCSIT}^{\rm AI\_SRP}(\cdot)$	 of the AI-SRP policy in \eqref{eq:BC-Rate-AoI-22}.

\begin{theorem}\label{thm:NoCSIT}
	In the   no CSIT case, when the broadcast strategy is adopted, the optimal objective value of the upper-bound problem in  \eqref{eq:upper-bound-problemBC},
	\begin{align*}
	U_{\rm NoCSIT}=R_{\rm NoCSIT}^{\rm AI\_SRP}(2\alpha-1), 
	\end{align*}
	where $R_{\rm NoCSIT}^{\rm AI\_SRP}(\cdot)$ is given by  \eqref{eq:BC-Rate-AoI-22}. 
\end{theorem}
\begin{proof}
	To prove the above theorem, we first argue that \eqref{eq:upper-bound-problemBC} must have a stationary randomized policy (SRP) as the optimal policy. We then show that the most general SRP one can employ in this case is similar to the AI-SRP proposed in the previous subsection. 
See Appendix E for the details. 	
\end{proof}

\subsection{Multiplicative and Additive Bounds on AI-SRP}
In the following theorem, we present multiplicative and additive bounds on the AI-SRP. 

\begin{theorem}\label{thm:rpBC}
Under the broadcast strategy, 	the following bounds hold true:
	\begin{itemize}
		\item For any $(R_0, H, \alpha,\bar{P})$, we have, $R^*_{\rm NoCSIT}\leq 2R_{\rm NoCSIT}^{\rm AI\_SRP}$. 
		\item $R^*_{\rm NoCSIT}-R_{\rm NoCSIT}^{\rm AI\_SRP}\leq  \nu_{\alpha}(\alpha-1)$, where $\nu_{\alpha}$ is the optimal dual variable corresponding to the average AoI constraint in \eqref{eq:outageBC}.
	\end{itemize}
\end{theorem}
\begin{proof}
The above results follow  from the concavity of the optimal long-term average throughput in $\alpha$ and they can be proved along the lines in the proof of Theorem~\ref{thm:rp}.  
\end{proof}

\section{Numerical Results}\label{sec: numerical results}
In this section,  we obtain simulation results, for which we consider an exponentially distributed channel power gain $H$ with unit mean with the PDF,  $f_h = e^{-h}$.  We truncate  $f_h$ at $h=h_{\rm max}$ and quantize $h$ to $N$ evenly spaced levels in $[0,h_{\rm max}]$, obtaining $h_i=ih_{\rm max}/N$ with probability $p_i=\int_{h=({i-1})h_{\rm max}/N}^{i h_{\rm max}/N}f_hdh$ for $i=1,\ldots,N-1$, $h_N=h_{\rm max}$ and  $p_N=\int_{h=({N-1})h_{\rm max}/N}^{\infty}f_hdh$. We let $h_{\rm max}=5$ for all the plots. 

\begin{figure*}[t]
	\centering
	\begin{subfigure}{0.48\textwidth}
%
%
\definecolor{mycolor1}{rgb}{0.00000,0.44700,0.74100}%
\definecolor{mycolor2}{rgb}{0.85000,0.32500,0.09800}%
\begin{tikzpicture}[scale=0.5]

\begin{axis}[%
width=5.425in,
height=3.152in,
at={(0.91in,0.548in)},
scale only axis,
xmin=1.2,
xmax=3.2,
xminorticks=true,
xlabel style={font=\bfseries\color{white!15!black}},
xlabel={$\alpha$},
ymin=0.95,
ymax=1.25,
ylabel style={font=\bfseries\color{white!15!black}},
ylabel={$M\triangleq U_x/R_x^*$},
axis background/.style={fill=white},
xmajorgrids,
xminorgrids,
ymajorgrids,
legend style={legend cell align=left, align=left, draw=white!15!black}
]
\addplot [color=mycolor1, line width=1.5pt, mark=o, mark options={solid, mycolor1},mark size = 4pt]
  table[row sep=crcr]{%
1.19935394620923	1.23120745078889\\
1.27427498570313	1.10766021139163\\
1.35387618002254	1.0530447645413\\
1.43844988828766	1.0212446127502\\
1.52830673265877	1.00954255065303\\
1.62377673918872	1.00278575245308\\
1.72521054994204	1.00073860141266\\
1.83298071083244	1.00010658508718\\
1.94748303990876	0.999999998168409\\
2.06913808111479	1.00000000179732\\
2.19839264886229	0.999999982173546\\
2.33572146909012	0.999999983593328\\
2.48162892283683	1.00000000738378\\
2.63665089873036	1.00000001412566\\
2.80135676119887	0.999999996786441\\
2.97635144163132	1.00000001160685\\
3.16227766016838	0.999999998219874\\
};
\addlegendentry{With CSIT, $\bar{P}=1,~R_0=1,\; N=50$}

\addplot [color=mycolor2, line width=1.5pt, mark=asterisk, mark options={solid, mycolor2}, mark size = 4pt]
  table[row sep=crcr]{%
1.27427498570313	1.13868450399839\\
1.35387618002254	1.05131763840591\\
1.43844988828766	1.01876908396042\\
1.52830673265877	1.00631365894461\\
1.62377673918872	1.00167650598641\\
1.72521054994204	1.00000000946637\\
1.83298071083244	1.00000001124829\\
1.94748303990876	1.00000001386022\\
2.06913808111479	1.00000000428868\\
2.19839264886229	1.00000000690131\\
2.33572146909012	0.999999994733004\\
2.48162892283683	1.00000001207785\\
2.63665089873036	0.999999999419563\\
2.80135676119887	1.00000000025071\\
2.97635144163132	1.0000000006903\\
3.16227766016838	1.00000000478261\\
};
\addlegendentry{No CSIT, $\bar{P}=5,~R_0=1,\; N=50$}

\end{axis}
\end{tikzpicture}%
		\caption{Multiplicative bound: $R_{x}^{\rm AI\_SRP}\geq R^*_x/M$, $x\in \{\rm CSIT, NoCSIT\}$.}
		\label{fig:CR1}
	\end{subfigure} 
	\hfill
	\begin{subfigure}{0.48\textwidth}
%
%
\definecolor{mycolor1}{rgb}{0.00000,0.44700,0.74100}%
\definecolor{mycolor2}{rgb}{0.85000,0.32500,0.09800}%
\begin{tikzpicture}[scale=0.5]

\begin{axis}[%
width=5.425in,
height=3.152in,
at={(0.91in,0.548in)},
scale only axis,
xmin=1.27,
xmax=3.2,
xlabel style={font=\bfseries\color{white!15!black}},
xlabel={$\alpha$},
ymin=0,
ymax=0.16,
ylabel style={font=\bfseries\color{white!15!black}},
yticklabel style={
	/pgf/number format/fixed,
	/pgf/number format/precision=5
},
ylabel={$A\triangleq\nu_{\alpha}(\alpha-1)$},
axis background/.style={fill=white},
xmajorgrids,
ymajorgrids,
legend style={legend cell align=left, align=left, draw=white!15!black}
]
\addplot [color=mycolor1, line width=1.5pt, mark=o, mark options={solid, mycolor1}, mark size = 4pt]
  table[row sep=crcr]{%
1.94748303990876	0.157414127178749\\
2.04421380004495	0.148121872065977\\
2.14574914115299	0.141373322055233\\
2.25232770498738	0.0662510373159511\\
2.36419998655008	0.066243203326811\\
2.48162892283683	0.0231512623839274\\
2.60489051082643	0.0234701348160011\\
2.73427445616523	0.00302643904459732\\
2.87008485407157	0.00309565831009184\\
3.01264090406045	2.74491540608324e-09\\
3.16227766016838	8.80673756142869e-10\\
};
\addlegendentry{With CSIT, $\bar{P}=1,\; R_0=1,\; N=50$.}

\addplot [color=mycolor2, line width=1.5pt, mark=asterisk, mark options={solid, mycolor2}, mark size = 4pt]
  table[row sep=crcr]{%
1.32145576997553	0.0193575589529171\\
1.38709198785093	0.0172125306298847\\
1.45598833231919	0.0156355722265533\\
1.52830673265877	0.0143936108439107\\
1.60421716111532	0.0134468843582209\\
1.68389803239289	0.0127108436814982\\
1.76753662298767	0.0121311411931124\\
1.85532951134996	0.0116780826188254\\
1.94748303990876	0.011327188443782\\
2.04421380004495	0.0110592324279493\\
2.14574914115299	0.0108590320638502\\
2.25232770498738	8.27923711810286e-08\\
2.36419998655008	4.93428163217402e-07\\
2.48162892283683	2.97843832708367e-07\\
2.60489051082643	1.58698908681032e-07\\
2.73427445616523	1.81465198423325e-08\\
2.87008485407157	1.20734079178675e-07\\
3.01264090406045	3.82382587815755e-08\\
3.16227766016838	4.14073126897563e-08\\
};
\addlegendentry{No CSIT, $\bar{P}=1.75,\; R_0=0.5,\; N=10$.}

\end{axis}
\end{tikzpicture}%
		\caption{Additive bound: $R_x^{\rm AI\_SRP}\geq R^*_x-A$, $x\in \{\rm CSIT, NoCSIT\}$.}
		\label{fig:CR2}
	\end{subfigure}	
	\caption{The variation of multiplicative and additive bounds on the optimal objective values, $R_{\rm CSIT}^{\rm AI\_SRP}$ and $R_{\rm NoCSIT}^{\rm AI\_SRP}$,  of the proposed AI-SRPs under the perfect and no CSIT cases, respectively, with the average AoI, $\alpha$.}
	\label{fig:bounds}
\end{figure*}
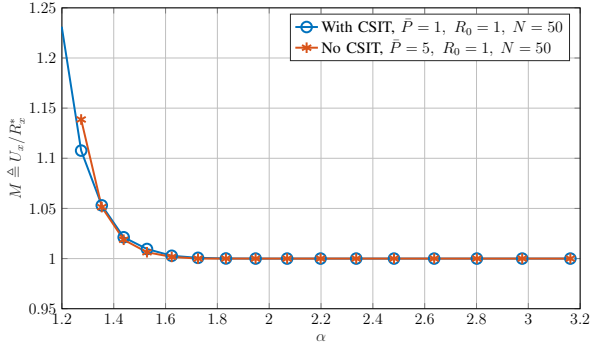
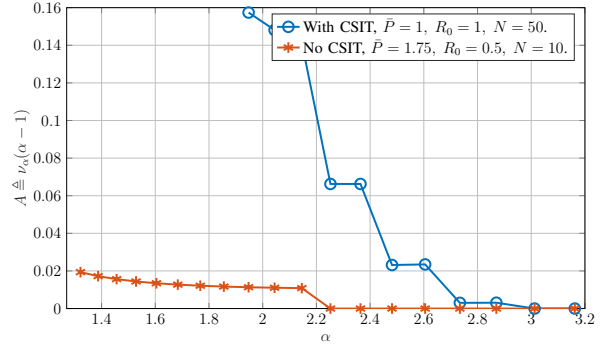 

\begin{figure*}[b!]
	\centering
	\begin{subfigure}{0.48\textwidth}
%
%
\definecolor{mycolor1}{rgb}{0.00000,0.44700,0.74100}%
\definecolor{mycolor2}{rgb}{0.85000,0.32500,0.09800}%
\begin{tikzpicture}[scale=0.5]

\begin{axis}[%
width=5.425in,
height=3.152in,
at={(0.91in,0.548in)},
scale only axis,
xmin=0.1,
xmax=2,
xlabel style={font=\bfseries\color{white!15!black}},
xlabel={$R_0$},
ymin=0.98,
ymax=1.08,
ylabel style={font=\bfseries\color{white!15!black}},
ylabel={$M\triangleq U_x/R_x^*$},
axis background/.style={fill=white},
xmajorgrids,
ymajorgrids,
legend style={legend cell align=left, align=left, draw=white!15!black}
]
\addplot [color=mycolor1, line width=1.5pt, mark=o, mark options={solid, mycolor1}]
  table[row sep=crcr]{%
0.1	1.00000000228532\\
0.2	0.999999998304134\\
0.3	0.999999999551542\\
0.4	0.999999996002368\\
0.5	1.00000000244213\\
0.6	1.00000000191303\\
0.7	1.00000001311281\\
0.8	0.999999997223119\\
0.9	1.0000000243137\\
1	0.999999999257992\\
1.1	0.99999998260576\\
1.2	1.00000000033299\\
1.3	1.0000069958484\\
1.4	1.00030691220201\\
1.5	1.00151782778485\\
1.6	1.00446738434646\\
1.7	1.01022798283791\\
1.8	1.02012233396192\\
1.9	1.03629239170427\\
2	1.06275268784049\\
};
\addlegendentry{With CSIT, $\bar{P}=1$, $\alpha=4$}

\addplot [color=mycolor2, line width=1.5pt, mark=asterisk, mark options={solid, mycolor2}]
  table[row sep=crcr]{%
	0.1	1.00000000343224\\
	0.2	0.999999994914742\\
	0.3	1.00000000440417\\
	0.4	0.999999993575354\\
	0.5	1.00000000137396\\
	0.6	0.999999999364063\\
	0.7	0.999999999898316\\
	0.8	1.0000000045877\\
	0.9	1.00009969218614\\
	1	1.00080745554195\\
	1.1	1.00227199698358\\
	1.2	1.00485461201081\\
	1.3	1.00926835026724\\
	1.4	1.01632307988311\\
	1.5	1.02750318735711\\
	1.6	1.04555676840439\\
};
\addlegendentry{No CSIT, $\bar{P}=10$, $\alpha=1.5$}

\end{axis}
\end{tikzpicture}%
		\caption{Multiplicative bound: $R_{x}^{\rm AI\_SRP}\geq R^*_x/M$, $x\in \{\rm CSIT, NoCSIT\}$.}
		\label{fig:CR11}
	\end{subfigure} 
	\hfill
	\begin{subfigure}{0.48\textwidth}
%
%
\definecolor{mycolor1}{rgb}{0.00000,0.44700,0.74100}%
\definecolor{mycolor2}{rgb}{0.85000,0.32500,0.09800}%
\begin{tikzpicture}[scale=0.5]

\begin{axis}[%
width=5.425in,
height=3.152in,
at={(0.91in,0.548in)},
scale only axis,
xmin=0.1,
xmax=1.5,
xlabel style={font=\bfseries\color{white!15!black}},
xlabel={$R_0$},
xtick={0.2, 0.4, 0.6, 0.8, 1.0, 1.2, 1.4},
ymin=0,
ymax=0.1,
yticklabel style={
	/pgf/number format/fixed,
	/pgf/number format/precision=5
},
ylabel={$A\triangleq\nu_{\alpha}(\alpha-1)$},
axis background/.style={fill=white},
xmajorgrids,
ymajorgrids,
legend style={legend cell align=left, align=left, draw=white!15!black}
]
\addplot [color=mycolor1, line width=1.5pt, mark=o, mark options={solid, mycolor1}, mark size = 4pt]
  table[row sep=crcr]{%
0.1	1.37175604209006e-10\\
0.2	2.23504104113204e-10\\
0.3	2.2169244218162e-10\\
0.4	1.73612679787993e-10\\
0.5	1.42657885504605e-10\\
0.6	1.23431265208751e-10\\
0.7	2.94267499256762e-10\\
0.8	1.27871935262647e-10\\
0.9	4.9656656564423e-10\\
1	4.35282032640316e-10\\
1.1	5.0402815254813e-10\\
1.2	3.82849529856344e-10\\
1.3	0.000665651822245206\\
1.4	0.025088939927852\\
1.5	0.0932424955918403\\
};
\addlegendentry{With CSIT, $\bar{P}=1$, $\alpha=4$}

\addplot [color=mycolor2, line width=1.5pt, mark=asterisk, mark options={solid, mycolor2}, mark size = 4pt]
  table[row sep=crcr]{%
0.1	1.12643228078468e-11\\
0.2	6.36668495701542e-12\\
0.3	2.25408580689646e-11\\
0.4	4.0959236002891e-11\\
0.5	2.02897698642346e-11\\
0.6	1.70279346178859e-11\\
0.7	2.33788544079516e-11\\
0.8	2.99278823945315e-10\\
0.9	0.00212215323585285\\
1	0.00628976391803726\\
1.1	0.0110015335410694\\
1.2	0.0200597191085612\\
1.3	0.0314890101485346\\
1.4	0.0503302847150957\\
1.5	0.0781752493511749\\
};
\addlegendentry{No CSIT, $\bar{P}=10$, $\alpha=1.5$}

\end{axis}
\end{tikzpicture}%
		\caption{Additive bound: $R_x^{\rm AI\_SRP}\geq R^*_x-A$, $x\in \{\rm CSIT, NoCSIT\}$.}
		\label{fig:CR21}
	\end{subfigure}	
	\caption{The variation of multiplicative and additive bounds on the optimal objective values, $R_{\rm CSIT}^{\rm AI\_SRP}$ and $R_{\rm NoCSIT}^{\rm AI\_SRP}$,  of the proposed AI-SRPs under the perfect and no CSIT cases, respectively, with $R_0$.}
	\label{fig:bounds1}
\end{figure*}
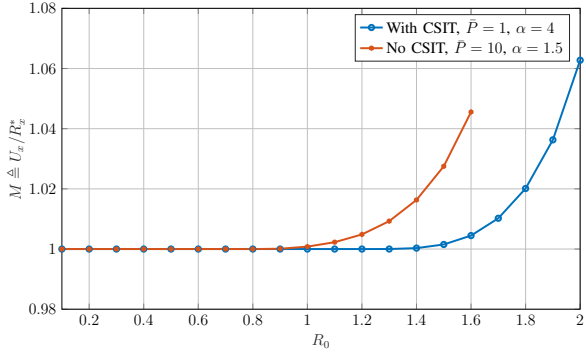
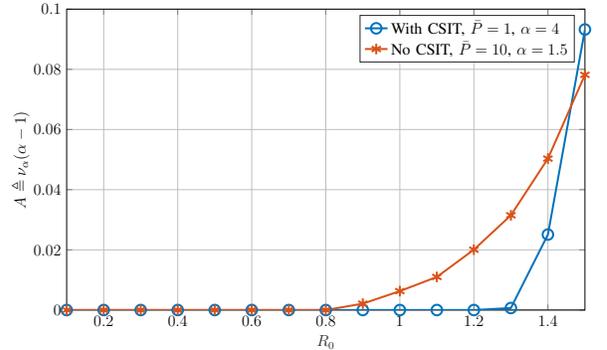 
In Fig.~\ref{fig:bounds} and Fig.~\ref{fig:bounds1}, we plot the variation  of the multiplicative and additive bounds on the performance gaps of the proposed AI-SRPs from the optimal performance under the perfect and no CSIT cases with $\alpha$ and $R_0$, respectively.  
From Theorem~\ref{thm:rp}, recall that $R_x^{\rm AI\_SRP}\geq R^*_x/2$ and $R_x^{\rm AI\_SRP}\geq R^*_x-\nu_{\alpha}(\alpha-1)$ for $x\in \{\rm CSIT, NoCSIT\}$. However, from Fig.~\ref{fig:CR1} and Fig.~\ref{fig:CR11} we note that,  for the selected parameters, the multiplicative bound is much smaller than $2$ both in the perfect and no CSIT cases. 
Further, from Fig.~\ref{fig:CR2} and Fig.~\ref{fig:CR21}, it can be seen that $R_x^*$ is only a fraction of bits away from the corresponding optimal performance for $x\in \{\rm CSIT, NoCSIT\}$. From the figures, it is also interesting to note that as $\alpha$ increases ($R_0$ decreases), the  gap between performance of the  proposed and the optimal policy decreases, and for a sufficiently large  $\alpha$ (small   $R_0$), the gap becomes zero.

 In Fig.~\ref{fig:CR12} and Fig.~\ref{fig:CR22}, we plot the variation of the long-term average throughput with average power, $\bar{P}$ in the proposed AI-SRPs under the perfect and no CSIT  cases, respectively. The long-term average throughput is a concave function of $\bar{P}$. This is because, the AI-SRPs proposed in \eqref{eq:stat_rand_CSIT_convex} and \eqref{eq:BC-Rate-AoI-22} for the CSIT and no CSIT cases, respectively, are convex optimization problems and hence, they are concave functions of $\bar{P}$ (see Exercise 5.32 in \cite{Boyd}). Further, as expected,  the long-term average throughput is non-increasing with the decreasing $\alpha$ and increasing $R_0$.  
This is because, for a smaller $\alpha$ and larger $R_0$, more power needs to be allocated to  lower channel power gain realizations, to satisfy the average AoI constraint, and this results in a lower throughput, due to the concavity of the rate function, $r(\cdot)$.

\section{Conclusions}\label{sec:conclusions}
We considered long-term average throughput maximization problems subject to average AoI and power constraints in a single user fading channel, when  perfect and no CSIT is available.  
We proposed simple \emph{age-independent} stationary randomized policies (AI-SRP),   whose optimal throughputs are at least equal to the half of the respective optimal long-term average throughputs, independent of all parameters of the problem. Moreover, the optimal throughputs are  within additive gaps from the respective optimal long-term average throughputs, where we express the additive gaps in terms of the optimal dual variable corresponding to their average AoI constraints.    We also obtain efficient algorithms to  solve the proposed AI-SRPs.  In summary, we found that \emph{age-independent} policies are \emph{near-optimal} in maximizing long-term average throughputs subject to average AoI and power constraints in fading channels. 

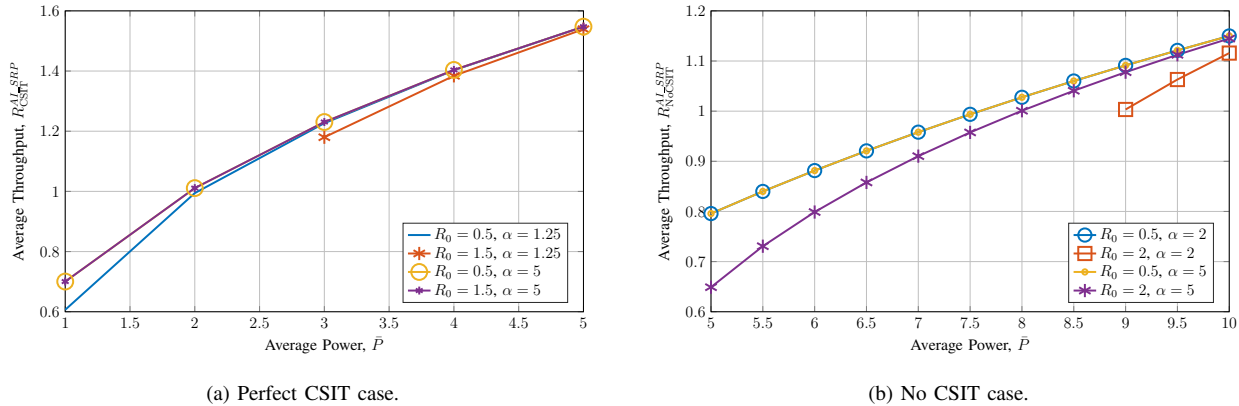
\begin{figure*}[t]
	\centering
	\begin{subfigure}{0.48\textwidth}
%
%
\definecolor{mycolor1}{rgb}{0.00000,0.44700,0.74100}%
\definecolor{mycolor2}{rgb}{0.85000,0.32500,0.09800}%
\definecolor{mycolor3}{rgb}{0.92900,0.69400,0.12500}%
\definecolor{mycolor4}{rgb}{0.49400,0.18400,0.55600}%
\begin{tikzpicture}[scale=0.5]

\begin{axis}[%
width=5.425in,
height=3.152in,
at={(0.91in,0.548in)},
scale only axis,
xmin=1,
xmax=5,
xlabel={Average Power, $\bar{P}$},
ymin=0.6,
ymax=1.6,
ylabel={Average Throughput, $R_{\rm CSIT}^{AI\_SRP}$},
axis background/.style={fill=white},
xmajorgrids,
ymajorgrids,
legend style={at={(0.97,0.03)}, anchor=south east, legend cell align=left, align=left, draw=white!15!black}
]
\addplot [color=mycolor1, line width=1.5pt, mark=, mark options={solid, mycolor1}, mark size = 5pt]
  table[row sep=crcr]{%
1	0.606514846676312\\
2	0.994286179183542\\
3	1.22663636689937\\
4	1.40263960363889\\
5	1.54741345225281\\
};
\addlegendentry{$R_0=0.5$, $\alpha=1.25$}

\addplot [color=mycolor2, line width=1.5pt, mark=asterisk, mark options={solid, mycolor2}, mark size = 5pt]
  table[row sep=crcr]{%
3	1.18000438510233\\
4	1.38413128046366\\
5	1.53918774787225\\
};
\addlegendentry{$R_0=1.5$, $\alpha=1.25$}

\addplot [color=mycolor3, line width=1.5pt, mark=o, mark options={solid, mycolor3}, mark size = 6pt]
  table[row sep=crcr]{%
1	0.70033432828656\\
2	1.01091792773527\\
3	1.23056428587547\\
4	1.40368935246076\\
5	1.54751920849338\\
};
\addlegendentry{$R_0=0.5$, $\alpha=5$}

\addplot [color=mycolor4, line width=1.5pt, mark=asterisk, mark options={solid, mycolor4}, mark size = 3pt]
  table[row sep=crcr]{%
1	0.700334334430035\\
2	1.01091793347428\\
3	1.23056430132743\\
4	1.40368937686571\\
5	1.54751919518438\\
};
\addlegendentry{$R_0=1.5$, $\alpha= 5$}

\end{axis}
\end{tikzpicture}%
		\caption{Perfect CSIT case.}
		
		\label{fig:CR12}
	\end{subfigure} 
	\hfill
	\begin{subfigure}{0.48\textwidth}
%
%
\definecolor{mycolor1}{rgb}{0.00000,0.44700,0.74100}%
\definecolor{mycolor2}{rgb}{0.85000,0.32500,0.09800}%
\definecolor{mycolor3}{rgb}{0.92900,0.69400,0.12500}%
\definecolor{mycolor4}{rgb}{0.49400,0.18400,0.55600}%
\begin{tikzpicture}[scale=0.5]

\begin{axis}[%
width=5.425in,
height=3.152in,
at={(0.91in,0.548in)},
scale only axis,
xmin=5,
xmax=10,
xlabel={Average Power, $\bar{P}$},
ymin=0.6,
ymax=1.2,
ylabel={Average Throughput, $R_{\rm NoCSIT}^{AI\_SRP}$},
axis background/.style={fill=white},
xmajorgrids,
ymajorgrids,
legend style={at={(0.97,0.03)}, anchor=south east, legend cell align=left, align=left, draw=white!15!black}
]
\addplot [color=mycolor1, line width=1.5pt, mark=o, mark options={solid, mycolor1}, mark size = 5pt]
  table[row sep=crcr]{%
5	0.795776762744678\\
5.5	0.839921493591936\\
6	0.881497855124596\\
6.5	0.920797098238189\\
7	0.958126729713365\\
7.5	0.993690240564288\\
8	1.02764721534949\\
8.5	1.0601365601379\\
9	1.09127988579062\\
9.5	1.12118441851522\\
10	1.14997490548248\\
};
\addlegendentry{$R_0=0.5$, $\alpha=2$}

\addplot [color=mycolor2, line width=1.5pt, mark=square, mark options={solid, mycolor2}, mark size = 5pt]
  table[row sep=crcr]{%
9	1.00335243661194\\
9.5	1.06308891597415\\
10	1.11581437160034\\
};
\addlegendentry{$R_0=2$, $\alpha=2$}

\addplot [color=mycolor3, line width=1.5pt, mark=o, mark options={solid, mycolor3}]
  table[row sep=crcr]{%
5	0.795776775749026\\
5.5	0.839921492331396\\
6	0.881497850456546\\
6.5	0.920797105763015\\
7	0.958126742346574\\
7.5	0.993690259942349\\
8	1.02764722301148\\
8.5	1.0601365597514\\
9	1.09127987325529\\
9.5	1.1211844187234\\
10	1.1499749079406\\
};
\addlegendentry{$R_0=0.5$, $\alpha=5$}

\addplot [color=mycolor4, line width=1.5pt, mark=asterisk, mark options={solid, mycolor4}, mark size = 5pt]
  table[row sep=crcr]{%
5	0.648940451805817\\
5.5	0.730672306677311\\
6	0.79882571535458\\
6.5	0.857872527081296\\
7	0.910290891982392\\
7.5	0.957594015495172\\
8	1.00083848552438\\
8.5	1.04067338859277\\
9	1.07772521851322\\
9.5	1.11238614538236\\
10	1.14482683670788\\
};
\addlegendentry{$R_0=2$, $\alpha=5$}

\end{axis}
\end{tikzpicture}%
		\caption{No CSIT case.}
		
		\label{fig:CR22}
	\end{subfigure}	
	\caption{The variation of the long-term average throughput of the proposed AI-SRP in the perfect and no CSIT cases with the maximum average power, $\bar{P}$.}
	\label{fig:bounds2}
\end{figure*}

\setcounter{subsection}{0}
\section*{Appendix}
\subsection{Proof of Proposition \ref{prop:UB}}
Consider a feasible policy, $\pi=(P(1),P(2),\ldots, P(K))$, over a finite horizon of $K$ slots. 
Now, for a given sample path of the channel gain process, $( {h}(1),\ldots, {h}(K))$, the total number of updates is equal to $M(k)\triangleq\sum_{k=1}^{K}u(k)$. For all $l\in \{1,\ldots,M(k)\}$, let the number of slots between $(l-1)^{\text{th}}$ and $l^{\text{th}}$ updates be denoted by $I[l]$.   
Let $S$ be the number of slots since the last update in $K$ slots. Then, we have, 
\begin{align}\label{eq:K}
K=\sum_{l=1}^{M(k)}I[l]+S. 
\end{align} 
In the slot following the $(l-1)^{\text{th}}$ update, the AoI is reset to $1$. Then, the AoI increases as $\{1,2,\ldots, I[l]\}$ until the $l^{\text{th}}$ update. Hence, the total AoI from the $(l-1)^{\text{th}}$ update to  the $l^{\text{th}}$ update is $I(l)(I(l)+1)/2$. 
Hence, the time average of AoI over $K$ slots can be expressed as
\begin{align}
\frac{1}{K}\sum_{k=1}^{K}a(k)&=\frac{1}{K}\left(\sum_{l=1}^{M(k)}\frac{I(l)(I(l)+1)}{2}+\frac{S(S+1)}{2}\right),\nonumber\\
&\stackrel{(a)}{=}\frac{1}{2}\left(\frac{M(k)}{K} \frac{1}{M(k)}\sum_{l=1}^{M(k)}I^2[l]+\frac{S^2}{K}+1\right),\nonumber\\ 
&\stackrel{(b)}{\geq}\frac{1}{2}\left(\frac{M(k)}{K} \left(\frac{1}{M(k)}\sum_{l=1}^{M(k)}I[l]\right)^2+\frac{S^2}{K}+1\right), \nonumber\\
&\stackrel{(c)}{=}\frac{1}{2}\left(\frac{1}{K} \frac{\left(K-S\right)^2}{M(k)}+\frac{S^2}{K}+1\right),\label{eq:KStoK}\\
&\stackrel{(d)}{\geq}\frac{1}{2}\left(\frac{K}{M(k)+1}+1\right),\label{eq:sp} 
\end{align}
where we substitute the linear terms in the above equation by \eqref{eq:K} for (a), (b) is because of the Cauchy–Schwartz inequality,  $\left(\frac{1}{n}\sum_{i=1}^{n}a\right)^2\leq \frac{1}{n}\sum_{i=1}^{n}a^2$ for $a_1,\ldots,a_n\geq 0$, (c) is obtained by replacing  $\sum_{l=1}^{M(k)}I[l]$ by $K-S$ from \eqref{eq:K}, and (d) is obtained by analytically minimizing \eqref{eq:KStoK} with respect to $S$. 

Taking expectation of \eqref{eq:sp}, we have,
\begin{align}
&\frac{1}{K}\mathbb{E}_{ H_1,\ldots,H_K}\left[\sum_{k=1}^{K}a(k)\right],\nonumber\\
&\geq \frac{1}{2}\mathbb{E}_{ H_1,\ldots,H_K}\left[\frac{K}{M(k)+1}+1\right],\nonumber\\
&\stackrel{(a)}{\geq }\frac{1}{2}\frac{K}{\mathbb{E}_{ H_1,\ldots,H_K}\left[M(k)\right]+1}+1,\nonumber\\
&\stackrel{}{=}\frac{1}{2}\frac{1}{\frac{1}{K}\mathbb{E}_{ H_1,\ldots,H_K}\left[M(k)\right]+\frac{1}{K}}+1,\nonumber\\
&\stackrel{(b)}{=}\frac{1}{2}\frac{1}{\frac{1}{K}\sum_{k=1}^{K}\mathbb{E}_H\left[u(k)\right]+\frac{1}{K}}+1\label{eq:aineq},
\end{align}
where (a) is due to Jensen's inequality and (b) is obtained by using the definition of $M(k)$. 

Now, applying the limit $K\rightarrow \infty$ to \eqref{eq:aineq}, we have, 
\begin{align}
&\lim_{K\rightarrow \infty} \frac{1}{K}\mathbb{E}_{ H_1,\ldots,H_K}\left[\sum_{k=1}^{K}a(k)\right],\nonumber\\
&\geq \lim_{K\rightarrow \infty} \frac{1}{2}\frac{1}{\frac{1}{K}\sum_{k=1}^{K}\mathbb{E}_{H}\left[u(k)\right]+\frac{1}{K}}+1\nonumber\\
&= \frac{1}{2}\frac{1}{\lim_{K\rightarrow \infty}\frac{1}{K}\sum_{k=1}^{K}\mathbb{E}_{ H}\left[u(k)\right]}+1. 
\end{align}
Now, applying the constraint on the long-term average age, and re-arranging the terms, we get the constraint, \eqref{eq:NofUpdates}. Since we are constraining a lower bound on the long-term average age, the optimal solution with \eqref{eq:NofUpdates} is an upper-bound to the original problem in \eqref{eq:main-opt-problem}.

\subsection{Proof of Theorem~\ref{thm:UB-CSIT}}
The upper-bound problem in  \eqref{eq:upper-bound-problem} has the structure of a constrained MDP \cite{VSharma,CMDP}, with state space $\mathcal{H}$ and action space  $\mathbb{R}^+\cup \{0\}$. Since the channel power gains are i.i.d. across frames, we have, $f_h(h''|h',P)=f_h(h'')$ for $h'',h'\in \mathcal{H}$, i.e., the next state does not depend on the current action or the state.  
When action $P(k)$ is taken, we define two  instantaneous rewards, namely $R_s(k)$ and $u(k)$, and an instantaneous cost, $P(k)$. 
Now, the optimization problem in  \eqref{eq:upper-bound-problem} can be transformed into the following unconstrained problem by the Lagrangian method:
\begin{subequations} \label{eq:upper-bound-problem1}
	\begin{align}
	\underset{P(1), P(2), \ldots}{\text{min}} &\;\lim_{K\rightarrow \infty}-\frac{1}{K}\sum_{k=1}^{K}\mathbb{E}[R_s(k)]-\nu \lim_{K\rightarrow \infty}\frac{1}{K}\sum_{k=1}^{K}\mathbb{E}[u(k)] +\lambda
	\lim_{K\rightarrow \infty}\frac{1}{K}\sum_{k=1}^{K}\mathbb{E}[P(k)], \;\;&\\
	\text{subject to}&\;\; P(k)\geq 0,\;\forall k\in\{1,2,\ldots\},\;\;
	\end{align}
\end{subequations}
where $\lambda, \nu\geq 0$ are Lagrange multipliers.
We now have the below result on the structure of the optimal solution to \eqref{eq:upper-bound-problem1}.

\begin{lemma}
	An SRP is optimal for solving \eqref{eq:upper-bound-problem1}.
\end{lemma}

\begin{proof}
	For any given $\lambda$ and $\nu$, \eqref{eq:upper-bound-problem1}   is  an infinite horizon expected average reward MDP. 
	In this case,  it will have an SRP as the optimal policy, if the conditions W1-W3 in  \cite{SRP-NEW} hold for this problem.  
	We now present W1-W3 (in italics) and argue that \eqref{eq:upper-bound-problem1}    satisfies these conditions in the following. 
	\begin{itemize}
		\item[W1.] \emph{$\mathcal{H}$ is a locally compact space with a countable base.} \\
		This condition holds  because $\mathcal{H}$ is a subset of $\mathbb{R}$, which is locally compact  with a countable base. 	 
		
		\item[W2.] \emph{The one-step cost function, $-R_s(k)-\nu u(k)+\lambda P(k)$ must  be inf-compact, i.e., for all $\beta\in \mathbb{R}$, the set $D(\beta) = \{(h(k), P(k)) \mid -R_s(k)-\nu u(k)+\lambda P(k) \leq \beta\}$ must be compact.} \\
		This condition holds because, for any $\beta\in \mathbb{R}$, $\nu,\lambda\geq 0$, $R_0\geq 0$ and $h(k)\in \mathcal{H}$,  we can find $P_0<\infty$ such that 
		$-R_s(k)-\nu u(k)+\lambda P(k) \leq \beta$ for all $P(k)\in [0,P_0]$, which is a compact set. Note that for some negative values of $\beta$ and  $\nu,\lambda\geq 0$, $R_0\geq 0$ and $h(k)\in \mathcal{H}$,  there may not exist a $P(k)\geq 0$ such that 	$-R_s(k)-\nu u(k)+\lambda P(k) \leq \beta$. In this case, the set $D(\beta)$ will be the null set, which is compact. 
		\item[W3.] \emph{Let $r(h'|h,P)$ denote the probability density of a transition from a state $h$ to the state $h'$.  Then, $r(h'|h,P)$ must be weakly continuous. That is, for any bounded and continuous function, $l: \mathcal{H}\rightarrow \mathbb{R}$,
			\begin{align}\label{eq:weakCont}
			\int_{\mathcal{H}} l(h') r(h'|h,P)dh',
			\end{align}
			must be continuous in $(h,P)$.}\\
		This condition holds because of the following. Since $f_h(h'|h,P)=f_h(h')$, \eqref{eq:weakCont} evaluates to a constant value irrespective of the value of $h$ and $P$ and  a constant function  is continuous.   
	\end{itemize} 
	
	Since the conditions W1-W3  hold for \eqref{eq:upper-bound-problem1},   it must have an SRP as the optimal policy \cite{SRP-NEW}. 
\end{proof}

Now, for a given $\lambda, \nu\geq 0$, the following is the \emph{most general randomized power policy}  one can adopt: When the channel realization is $h$,  transmit with random power $P^{\rm rand}_h\geq 0$ with PDF, $g(P^{\rm rand}_h)\geq 0$, where $\int_{0}^{\infty}g(P^{\rm rand}_h) dP^{\rm rand}_h=1$.  The above policy is the most general SRP in the sense that any SRP is a special case of the above policy,  as the PDF, $g(P^{\rm rand}_h)$, is arbitrary. 
We now have the following important lemma.

\begin{lemma}
	For a given $\lambda, \nu\geq 0$,  when the channel power gain is $h$, for the most general random power policy, $P^{\rm rand}_h$, there exists a power policy of cardinality two  with the following structure:  Transmit with power
	\begin{align}\label{eq:UB-policy}
	P_h=  
	\begin{cases}
	P_h^1& \text{with probability}\; \mu_h,\\
	P_h^2           & \text{otherwise},   
	\end{cases}
	\end{align}
	where, 
	$	\mu_h  \triangleq \mathbb{P}\left(P^{\rm rand}_h>\frac{r^{-1}(R_0)}{h}\right)$, 
	$P_h^1> \frac{r^{-1}(R_0)}{h}$ and $P_h^2\leq \frac{r^{-1}(R_0)}{h}$,  
	such that 
	\begin{align}
	&\mathbb{E}_{f_h}\left[\mathbb{E}_g[P^{\rm rand}_h]\right]= \mathbb{E}_{f_h}\left[\mu_h {P}_h^1+(1-\mu_h ){P}_h^2\right],\label{eq: same-power}\\
	&\mathbb{E}_{f_h}\left[\mathbb{E}_g\left[r(hP^{\rm rand}_h)\right]\right]\leq \mathbb{E}_{f_h}\left[\mu_hr\left(h{P}^1_h\right) +(1-\mu_h )r\left(h{P}^2_h\right)\right].\label{eq:better rate} 
	\end{align}
	
\end{lemma}
\begin{proof}
	The result can be proved along the lines in the proof of Proposition 3.1 in \cite{Goldsmith}. Essentially, the result follows because of the concavity of $r(\cdot)$. 	
	Consider the fully randomized policy, $P^{\rm rand}_h$ and define the following deterministic quantities for a given $h\in \mathcal{H}$:  
	\begin{align*}
	&\tilde{P}_h^1 \triangleq \mathbb{E}_g\left[P^{\rm rand}_h\mid P^{\rm rand}_h>\frac{r^{-1}(R_0)}{h}\right],\\
	&\tilde{P}_h^2 \triangleq \mathbb{E}_g\left[P^{\rm rand}_h\mid P^{\rm rand}_h\leq \frac{r^{-1}(R_0)}{h}\right]. 
	\end{align*}	
	Now, we note the following:  
	\begin{align*}
	\mathbb{E}_g[P^{\rm rand}_h]= & \mathbb{E}_g\left[P^{\rm rand}_h\mid P^{\rm rand}_h>\frac{r^{-1}(R_0)}{h}\right] \times\mathbb{P}\left(P^{\rm rand}_h>\frac{r^{-1}(R_0)}{h}\right)+\nonumber\\
	& \mathbb{E}_g\left[P^{\rm rand}_h\mid P^{\rm rand}_h\leq \frac{r^{-1}(R_0)}{h}\right]\times\mathbb{P}\left(P^{\rm rand}_h\leq \frac{r^{-1}(R_0)}{h}\right), \\
	= & \tilde{P}_h^1\mu_h+\tilde{P}_h^2(1-\mu_h),\\
	&\hspace{-2cm}\implies \mathbb{E}_{f_h}\left[\mathbb{E}_g[P^{\rm rand}_h]\right]= \mathbb{E}_{f_h}\left[\tilde{P}_h^1\mu_h\right]+\mathbb{E}_{f_h}\left[\tilde{P}_h^2(1-\mu_h)\right]. 
	\end{align*}
	This proves \eqref{eq: same-power}. 
	To prove \eqref{eq:better rate}, note that
	\begin{align}\label{eq:P1}
	\mathbb{E}_g\left[\log(1+hP^{\rm rand}_h)\mid P^{\rm rand}_h>\frac{r^{-1}(R_0)}{h}\right]&\stackrel{(a)}{\leq }r\left(h\mathbb{E}_g\left[P^{\rm rand}_h\mid P^{\rm rand}_h>\frac{r^{-1}(R_0)}{h}\right]\right)\nonumber\\
	&\stackrel{(b)}{=}r\left(h\tilde{P}^1_h\right),
	\end{align}
	where (a) is due to the Jensen's inequality and (b) follows from the definition of $\tilde{P}^1_h$. 
	Similarly, we have
	\begin{align}\label{eq:P2}
	&\mathbb{E}_g\left[\log(1+hP^{\rm rand}_h)\mid P^{\rm rand}_h\leq \frac{r^{-1}(R_0)}{h}\right]\leq r\left(h\tilde{P}^2_h\right).
	\end{align}
	From \eqref{eq:P1} and \eqref{eq:P2}, we get \eqref{eq:better rate}. 

\end{proof}

The above lemma says that in each fading state, for every fully general SRP, $P^{\rm rand}_h$ with some average power consumption,  there exists another policy having the structure in   \eqref{eq:UB-policy} that achieves the same or a higher throughput, with the same average power consumption as the fully general policy.  This implies that the optimal SRP that solves \eqref{eq:upper-bound-problem1}  has the structure of  \eqref{eq:UB-policy}. In the following, we apply \eqref{eq:UB-policy} and re-write \eqref{eq:upper-bound-problem1}. 

With \eqref{eq:UB-policy},   $\mathbb{E}[R_s(k)]$, $\mathbb{E}[u(k)]$ and $\mathbb{E}[P(k)]$ are given by \eqref{eq:objSRP}, and left-hand-sides of \eqref{eq:AOI} and \eqref{eq:Pc}, respectively. Defining $v^1_h\triangleq \mu_hP^1_h$ and $v^2_h\triangleq (1-\mu_h)P^2_h$, 
\eqref{eq:upper-bound-problem1}  can be re-written as 
\begin{subequations}\label{eq:stat_rand_CSIT_convex1}  
	\begin{align}
	\underset{\mu_h, v^1_h, v^2_h\geq 0}{\text{max}} &\;\;\int_{h\in \mathcal{H}} \left(\mu_hr\left(\frac{hv_h^1}{\mu_h}\right)+(1-\mu_h)r\left(\frac{hv_h^2}{1-\mu_h}\right)\right)f_hdh+\nu \int_{h\in \mathcal{H}}\mu_h f_hdh\nonumber\\
	&\;-\lambda\int_{h\in \mathcal{H}}   \left(v_h^1+v^2_h\right)f_hdh,\\
	\text{subject to}&\;   v_1^h\geq \mu_h\left(\frac{r^{-1}(R_0+\epsilon)}{h}\right),~ \mu_h	\leq 1, \; v_2^h\leq  (1-\mu_h) \left(\frac{r^{-1}(R_0)}{h}\right).
	\end{align}
\end{subequations} 
The above optimization problem is jointly convex and we can optimally solve it for various values of $\lambda, \nu\geq 0$. 
Now,     solving  \eqref{eq:stat_rand_CSIT_convex1} for a choice of $(\lambda,\nu)$ pair gives  the optimal solution to the original constrained optimization problem, if it results in  $\int_{h\in \mathcal{H}}\mu_h f_hdh\geq 1/(2\alpha-1)$ and $\int_{h\in \mathcal{H}}   \left(v_h^1+v^2_h\right)f_hdh\leq \bar{P}$, and the optimal objective function does not attain a higher value for any other choice of $(\lambda,\nu)$. 
In order to obtain the optimal $(\lambda,\nu)$, consider the following optimization problem. 
\begin{subequations}\label{eq:stat_rand_CSIT_convex2}  
	\begin{align}
	\underset{\mu_h, v^1_h, v^2_h\geq 0}{\text{max}} &\;\;\int_{h\in \mathcal{H}} \left(\mu_hr\left(\frac{hv_h^1}{\mu_h}\right)+(1-\mu_h)r\left(\frac{hv_h^2}{1-\mu_h}\right)\right)f_hdh,  \;\;\label{eq:obj2}\\
	\text{subject to}&\; \int_{h\in \mathcal{H}}\mu_h f_hdh\geq \frac{1}{2\alpha-1},\, \mu_h	\leq 1,\label{eq:outage2} \\
	&\;  \int_{h\in \mathcal{H}}   \left(v_h^1+v^2_h\right)f_hdh \leq \bar{P},\label{eq:avp2} \\
	&\; v_1^h\geq \mu_h\left(\frac{r^{-1}(R_0+\epsilon)}{h}\right),\; v_2^h\leq  (1-\mu_h) \left(\frac{r^{-1}(R_0)}{h}\right),
	\end{align}
\end{subequations} 
By inspection, we can see that \eqref{eq:stat_rand_CSIT_convex1} is a Lagrangian formulation of \eqref{eq:stat_rand_CSIT_convex2}, with $\lambda$ and $\nu$ being the Lagrange multipliers corresponding to \eqref{eq:avp2} and \eqref{eq:outage2}, respectively. Hence, the optimal $\lambda$ and $\nu$ for which \eqref{eq:stat_rand_CSIT_convex1} attains its maximum value are respectively the optimal Lagrange multipliers corresponding to \eqref{eq:avp2} and \eqref{eq:outage2}. 

We finally note that \eqref{eq:stat_rand_CSIT_convex2} is identical to \eqref{eq:stat_rand_CSIT_convex} except that the $\alpha$ in \eqref{eq:stat_rand_CSIT_convex} has been replaced by $2\alpha-1$ in \eqref{eq:stat_rand_CSIT_convex2}. Hence, we conclude that, $U_{\rm CSIT}=R_{\rm CSIT}^{\rm AI\_SRP}(2\alpha-1)$.

\subsection{ The Optimal Solution to \eqref{eq:LSC} with $r(x)=\log(1+x)$}
 When $r(x)=\log(1+x)$, an efficient algorithm, referred to as the layered water-filling (LWF) algorithm, has been obtained in \cite{NoCSIT} as the optimal solution to \eqref{eq:LSC}.  We present it here, as the algorithm and its structural properties are useful in obtaining an iterative algorithm for optimally solving the proposed AI-SRP,  when  $r(x)=\log(1+x)$.

When $r(x)=\log(1+x)$, \eqref{eq:uncon-avgP} becomes 
\begin{align*}
\left(\sum_{i=1}^{N}(h_i^{-1}-h_{i+1}^{-1}) \exp\left(\sum_{k=1}^{i}R_k\right)\right)-h_1^{-1}\leq P. 
\end{align*}
Since \eqref{eq:LSC} is convex, the KKT conditions are necessary and sufficient for optimality. Let $\lambda, \beta_i\geq 0$ be the Lagrange multipliers corresponding to \eqref{eq:uncon-avgP} and the inequality $R_i\geq 0$, respectively. Then, based on the KKT conditions, we can find that the optimal solution to \eqref{eq:LSC} admits the following: 
\begin{align}\label{eq:unconstrained1}
\lambda \exp\left(\sum_{m=1}^{i}R_m\right)=\frac{p_i}{h_i^{-1}-h^{-1}_{i+1}} + \frac{\beta_i-\beta_{i+1} }{h_i^{-1}-h_{i+1}^{-1}}, 
\end{align}
where $h_{i+1}\triangleq\infty$.  
Further, the complementary slackness condition requires $\beta_{i}R_{i}=0$. Hence, whenever $R_{i}>0$, we must have,  $\beta_{i}=0$.  
Moreover, based on the KKT conditions, an algorithm, referred to as the layered water-filling (LWF) algorithm, presented in Algorithm~\ref{algo:lwf} has been shown to be optimal in \cite{NoCSIT}. The LWF algorithm has the following two important characteristics:

\begin{itemize}
	\item Not every layer may be allocated a non-zero amount of power. 
	Let $\mathcal{A}=\{a_1, \ldots, a_L\}$, $L\leq N$, be the set of \emph{active} layers for which a non-zero power is allocated. Further, let $a_1, \ldots, a_L$ be such that  $h_{a_1}<h_{a_2}<\ldots h_{a_L}$.  
	The set $\mathcal{A}$ can be found as follows \cite{Diggavi,NoCSIT}:

	Assume that all the layers are active, i.e., let $\mathcal{A}=\{1,\ldots,N\}$. 
	If ${p_i}/\left({h_i^{-1}-h^{-1}_{i+1}}\right) \leq {p_{i-1}}/\left({h^{-1}_{i-1}-h^{-1}_{i}}\right)$ for any  $i\in \mathcal{A}$, then we must have $R_{i} \leq 0$ in order to satisfy \eqref{eq:unconstrained1}. Since $R_{i}$ cannot be negative, we must have, $R_{i}=0$. Noting this, we remove layer $i$ from $\mathcal{A}$ and assign $p_{i-1}=p_{i-1}+p_i$ as the probability mass of $h_{i-1}$. 
	We continue to merge the layers until ${p_{a_i}}/\left({h_{a_i}^{-1}-h^{-1}_{a_{i+1}}}\right)$ is strictly increasing with $i=1\ldots,L$, where we recall that $a_i$ is the $i^{\rm th}$ element in  $\mathcal{A}$ such that $h_{a_1}<\ldots<h_{a_{L}}$ and $L=|\mathcal{A}|$.  We summarize the above procedure for  computing $\mathcal{A}$ in steps 2-11 of Algorithm~\ref{algo:lwf}. Note that $\mathcal{A}$ depends only on the channel power gain distribution. 
	
		\begin{table*}[t]
		\begin{equation}\label{eq:pmax}
		P^{\mathrm{max}}_{a_l} =\left(\frac{p_{a_l}(h^{-1}_{a_{l-1}}-h^{-1}_{a_l})}{p_{a_{l-1}}(h^{-1}_{a_l}-h^{-1}_{a_{l+1}})}-1\right)\left(h^{-1}_{a_l}-h^{-1}_{a_{l+1}}+ \sum_{j=l+1}^{A}\left((h^{-1}_{a_j}-h^{-1}_{a_{j+1}})\exp\left(\sum_{i=l+1}^{j}\log\left(\frac{p_{a_i}(h^{-1}_{a_{i-1}}-h^{-1}_{a_i})}{p_{a_{i-1}}(h^{-1}_{a_i}-h^{-1}_{a_{i+1}})}\right)   \right)\right)	\right)
		\end{equation}	
	\end{table*}

	\item 	Among the active layers, power is allocated first to layer $a_L$, followed by the consecutive lower layers. The optimal power allocated to layer $a_l, \; l=1,\ldots,L$, is given by
	\begin{equation} \label{eq:optLSC}
	P^*_{a_l}=\left\{\,
	\begin{IEEEeqnarraybox}[][c]{l?s}
	\IEEEstrut
	P^{\mathrm{max}}_{a_l}   & if $P^{\mathrm{max}}_{a_l} \leq P -\sum_{j=l+1}^{L}P^{\mathrm{max}}_{a_j}$,\\
	P -\sum_{j=l+1}^{L}P^{\mathrm{max}}_{a_j}   & otherwise.
	\IEEEstrut
	\end{IEEEeqnarraybox}
	\right.
	\end{equation}
	where $P^{\mathrm{max}}_{a_1}=\infty$ and $P^{\mathrm{max}}_{a_l}$  for $l=2,\ldots,L$ is given by \eqref{eq:pmax}. Now, the corresponding rates can be computed from   \eqref{eq:rate_lsc}. 	We summarize the above procedure for  computing optimal transmit powers and rates in steps 12-15 of Algorithm~\ref{algo:lwf}. 
\end{itemize}

\begin{algorithm}[t]
	\caption{The layered water-filling (LWF) algorithm}
	\label{algo:lwf}
	{\small 
		\begin{algorithmic}[1]
			\Procedure{LWF}{$(h_1,\ldots,h_N), (p_1,\ldots,p_N), P$}
			\State $\mathcal{A}\coloneqq \{1,\ldots,N\}$,  $R_i,P_i\coloneqq 0$\, $\forall$\, $i\in \{1,\ldots,N\}$.  
			\State $a_i \coloneqq i^{\rm th}$ \text{element in} $\mathcal{A}$ such that $h_{a_1}<\ldots<h_{a_{|\mathcal{A}|}}$. 
			\State $L \coloneqq |\mathcal{A}|$. 
			\For {$i= 2$ to $L$}
			
			\If {$\frac{p_{a_i}}{h_{a_i}^{-1}-h^{-1}_{a_{i+1}}} \leq \frac{p_{a_{i-1}}}{h^{-1}_{a_{i-1}}-h^{-1}_{a_i}}$}
			\State  Update mass of $h_{a_{i-1}}$ as $p_{a_{i-1}}\coloneqq p_{a_{i-1}}+p_{a_i}$.
			\State Remove $a_i$ from $\mathcal{A}$, i.e., $\mathcal{A}\coloneqq \mathcal{A} \setminus \{a_i\}$. 
			\EndIf  
			
			\EndFor
			\State Repeat 3-10 until $\frac{p_{a_i}}{h_{a_i}^{-1}-h^{-1}_{a_{i+1}}} > \frac{p_{a_{i-1}}}{h^{-1}_{a_{i-1}}-h^{-1}_{a_i}} \, \forall \, i = 2,\ldots, L$. 
			\For {$j=a_L$ to $a_1$}
			\State $P_{a_l}\coloneqq \min \left(P^{\mathrm{max}}_{a_l}, P  -\sum_{i=l+1}^{L}P^{\mathrm{max}}_{a_i}\right)$, where  $P^{\mathrm{max}}_{a_l}$ is defined in \eqref{eq:pmax}. 
			\State Compute $R_i$ from \eqref{eq:rate_lsc}. 
			\EndFor 	
			\State Output $(R_1,\ldots,R_N)$.  		 
			\EndProcedure
	\end{algorithmic}}
\end{algorithm}

\subsection{Optimal Solution to \eqref{eq:BC-Rate-AoI-22} with $r(x)=\log(1+x)$}
In this subsection, we obtain an algorithm to solve \eqref{eq:BC-Rate-AoI-22} when $r(x)=\log(1+x)$. 
Since \eqref{eq:BC-Rate-AoI-22} is jointly convex, the  KKT conditions are necessary and sufficient for optimality. 
The Lagrangian of \eqref{eq:BC-Rate-AoI-22} is given by
\begin{align}
&L=\lambda \left(\sum_{j=1}^{N} \mu_j\left(\sum_{i=1}^{N}(h_i^{-1}-h_{i+1}^{-1}) \exp\left(\frac{\sum_{k=1}^{i}v_k^j}{\mu_j}\right)\right) 	-h_1^{-1}\right)\nonumber\\
&-\lambda  \bar{P}-\sum_{j=1}^{N}\sum_{i=1}^{N}q_iv_i^j-\delta\left(\sum_{j=1}^{N}q_j\mu_j-\alpha^{-1}\right)+\beta\sum_{j=1}^{N}\mu_j-\beta\nonumber\\
&-\sum_{j=1}^{N}\omega_j\left(\sum_{i=1}^{j}v_i^j- \mu_j(R_0+\epsilon)\right)-\sum_{j=1}^{N}\sum_{i=1}^{N}\nu_i^jv_i^j-\sum_{j=1}^{N}\psi_j\mu_j,
\end{align}
where $\lambda, \delta, \beta, \omega_j, \nu_i^j, \psi_j\geq 0$ are the Lagrange multipliers.   
The optimal solution must satisfy, $\frac{\partial L}{\partial v_i^j},\frac{\partial L}{\partial v_{i+1}^j}=0$ for all $i,j=1,\ldots,N$, from which we can obtain the following equation:  
\begin{subequations} 
	\begin{align}
	&\lambda \exp\left(\frac{\sum_{m=1}^{i}v_m^j}{\mu_j}\right)=\frac{p_i}{h_i^{-1}-h_{i+1}^{-1}}+\frac{\nu_i^j-\nu_{i+1}^j}{h_i^{-1}-h_{i+1}^{-1}}, \;\; \text{for}\;\; i=1,\ldots,j-1,j+1,\ldots,N,\\
	&\lambda \exp\left(\frac{\sum_{m=1}^{j}v_m^j}{\mu_j}\right)=\frac{p_j+\omega_j}{h_j^{-1}-h_{j+1}^{-1}}+\frac{\nu_j^j-\nu_{j+1}^j}{h_j^{-1}-h_{j+1}^{-1}}, 
	\end{align}	
\end{subequations} 
for all $j=1,\ldots,N$. 
Note that the above equations are  identical to  \eqref{eq:unconstrained1} when $p_j$ is  replaced by $p_j+\omega_j$.  
Suppose we know $\omega_j$ and $P^j\triangleq \sum_{i=1}^NP_i^j$, we can apply the LWF algorithm in Algorithm~\ref{algo:lwf} and obtain optimal rates $R_1^{j,\rm opt},\ldots, R_N^{j,\rm opt}$ as follows.

\paragraph{Obtaining Optimal Rates, $R_1^{j,\rm opt},\ldots, R_N^{j,\rm opt}$, for known $\omega_j$ for $j=1,\ldots,N$}
From the LWF algorithm, we can obtain the following: 
\begin{itemize}
	\item The set of active layers for the $j^{\rm th}$ rate tuple, which we denote by  $\{a_1^j,\ldots,a_{L_j}^j\}$  and  
	\item The maximum transmit power  and rate in the $i^{\rm th}$ active layer of the $j^{\rm th}$ rate-tuple, which we denote by $P_i^{j,\rm max}$ and $R_i^{j,\rm max}$, respectively, for $i = a_2^j, \ldots,a_{L_j}^j$ and $j=1,\ldots,L$. 
\end{itemize}

Given $P_i^{j,\rm max}$ and $R_i^{j,\rm max}$ for $i = a_2^j, \ldots,a_{L_j}^j$ and $j=1,\ldots,L$, in order to find the transmit power in the first layer of the $j^{\rm th}$ rate tuple, $P_1^j$, and the optimal transmit probability, $\mu_j$, we need to solve the following optimization problem:  
\begin{subequations}\label{eq:E1opt}
	\begin{align}
	\underset{E_1^j, \mu_j}{\text{max}} \;\sum_{j=1}^{N}\mu_j\tilde{q}_1^j&\log\left(1+\frac{\tilde{h}_1^jE_1^j} {\mu_j\left(1+\tilde{h}_1^j\sum_{k=2}^{L_j}P_k^{j,\rm max}\right)}\right)+\sum_{j=1}^{N}\mu_j\sum_{k=2}^{L_j}R_k^{j,\rm max}, \;\;&\\
	\text{subject to}&\;\;
	\sum_{j=1}^{N} \left(E_1^j+ \mu_j\sum_{k=2}^{L_j}P_k^{j,\rm max} \right)\leq \bar{P},\;\; \label{eq:avg-P}  \\
	&\;\; \sum_{j=1}^{N}q_j\mu_j\geq \frac{1}{\alpha},\; \sum_{j=1}^{N}\mu_j\leq 1,\\
	&\;\;  \sum_{i=1}^{j}v_i^j\geq \mu_j(R_0+\epsilon),\;\forall j=1,\ldots,N,  \\
	&\;\; \mu_j,  E_1^j\geq 0,\;\forall j=1,\ldots,N, 
	\end{align}
\end{subequations} 
where $E_1^j\triangleq P_1^j\mu_j$,  $\tilde{h}^j_i\triangleq h_{a_i^j}$, $\tilde{p}^j_i\triangleq p_{a_i^j}$ and $\tilde{q}_i^j=\sum_{k=i}^{L_j}\tilde{p}_k^j$ for all $i=1,\ldots, L_j$ and $j=1,\ldots,N$.

The above problem is jointly convex and we note the following.
\begin{itemize}
	\item For fixed values of $\mu_j$, using the KKT conditions, we can find that the optimal solution of the above problem is
	\begin{align}
	P_1^{j,\rm opt} =  \frac{1}{\delta}-\frac{1}{\tilde{h}_1^j} - \sum_{i=2}^{L_j}P_i^{j,\rm max}, \, \forall\, j=1,\ldots,N, 
	\end{align}
	where $\delta$ is chosen to satisfy \eqref{eq:avg-P} with equality. 
	\item Now, for the fixed values of $E_1^{j,\rm opt} = P_1^{j,\rm opt}\mu_j,\,\forall\, j=1,\ldots,N$, the above problem is convex and we can solve  it for obtaining optimal $\mu_j^*$ using standard numerical techniques. 
\end{itemize} 
From the above observations, it is clear that we can optimally solve \eqref{eq:E1opt} by repeating the above steps in  alternating manner until convergence to the global optimal solution, $E_1^{j,\rm opt},\mu_j^*$ for all $j=1,\ldots,N$. Then, the optimal $R_1^{j,\rm opt}=\log\left(1+ ({\tilde{h}_1^jE_1^{j,\rm opt}/\mu_j^*})/({1+\tilde{h}_1^j\sum_{k=2}^{L_j}P_k^{j,\rm max}})\right),\, \forall\, j=1,\ldots,N$.

\paragraph{Obtaining Optimal $\omega_j$ for known $R_1^{j,\rm opt},\ldots,R_N^{j,\rm opt}$ for $j=1,\ldots,N$}
In order to obtain the optimal $\omega_j$'s, consider the following dual function. 
\begin{align*}
g(\omega_1,\ldots,\omega_N)=&-\sum_{j=1}^{N}\sum_{i=1}^{N}q_iv_i^j-\sum_{j=1}^{N}\omega_j\left(\sum_{i=1}^{j}v_i^j- \mu_j(R_0+\epsilon)\right),
\end{align*}
where we recall that $v_i^j=\mu_jR_i^{j,\rm opt}$. 
From \cite{Boyd}, we note that $g(\omega_1,\ldots,\omega_N)$ must necessarily be concave, and the optimal $\omega_1^*,\ldots,\omega_N^*$ are given by
\begin{align}\label{eq:opt-dual-var}
(\omega_1^*,\ldots,\omega_N^*)=\argmax g(\omega_1,\ldots,\omega_N). 
\end{align}
The above problem can be efficiently solved using a sub-gradient descent method.

Finally, we can obtain the optimal solution to \eqref{eq:BC-Rate-AoI-22} by finding optimal $R_1^{j,\rm opt},\ldots,R_N^{j,\rm opt}$ for fixed  $\omega_j$ for $j=1,\ldots,N$ and vice versa,   until convergence. Since \eqref{eq:BC-Rate-AoI-22} is jointly convex, the above alternating optimization  will converge to the global optimal solution.

\subsection{Proof of Theorem~\ref{thm:NoCSIT}}
In order to solve the upper-bound problem in \eqref{eq:upper-bound-problemBC}, we note that since the CSIT is unavailable, the decision on transmit powers in each block is made using only the  channel power gain  distribution. Moreover, the action we take in a block does not impact the channel  power gain distribution in the subsequent fading blocks. Hence, the optimal action we take must be stationary (possibly random). In the following, we show that the most general SRP one can adopt to solve the upper-bound problem is identical to the AI-SRP proposed in Section~\ref{eq:AISRP-NoCSIT}.

The most general SRP under the broadcast strategy is the following: At each slot, a rate-tuple $(R_1,\ldots,R_N)$, with corresponding power-tuple $(P_1,\ldots,P_N)$, is chosen with an arbitrary joint PDF, $g(R_1,\ldots,R_N)$. The above policy is the most general SRP in the sense that any SRP is a special case of the above policy, as the PDF, $g(R_1,\ldots,R_N)$, is arbitrary. 
Then, $N$ codewords with rates $R_1,\ldots,R_N$ are selected, superimposed and transmitted. 
Define   
	\begin{align}\label{eq:muj}
	\mu_j\triangleq \mathbb{P}\left(\sum_{i=1}^jR_i>R_0, \sum_{i=1}^{j-1}R_i\leq R_0\right),
	\end{align}  
	for $j=1,\ldots,N$. 
	Now, in the broadcast strategy, conditioned on the event that  the transmitted rate-tuple  satisfies $\sum_{i=1}^jR_i>R_0$ and $\sum_{i=1}^{j-1}R_i\leq R_0$,  it will be successful whenever $H\geq h_j$, i.e.,  with probability $q_j$. 
	Hence, the success probability is given by $\sum_{j=1}^{N}\mu_jq_j$. 
	Further, define   
		\begin{align}\label{eq:Rij}
 \mathcal{G}^j\triangleq\left\{\sum_{k=1}^jR_k>R_0, \sum_{k=1}^{j-1}R_k\leq R_0\right\}	\;\; \text{and}\;\; R_i^j \triangleq  \mathbb{E}_g\left[R_i \mid \mathcal{G}^j\right],   
	\end{align} 
	for $j=1,\ldots,N$. 
	Now, the long-term average throughput achievable by the above policy is   $\mathbb{E}_g\left[\sum_{i=1}^Nq_iR_i\right]$, and we have 
	\begin{align}\label{eq:MostGenRate}
	\mathbb{E}_g\left[\sum_{i=1}^Nq_iR_i\right]&=\sum_{j=1}^{N}\mu_j\mathbb{E}_g\left[\sum_{i=1}^Nq_iR_i \mid \mathcal{G}^j\right],\nonumber\\
	&\stackrel{(a)}{=}\sum_{j=1}^{N}\mu_j\sum_{i=1}^Nq_i\mathbb{E}_g\left[R_i \mid \mathcal{G}^j\right]\stackrel{(b)}{=}\sum_{j=1}^{N}\mu_j \sum_{i=1}^Nq_iR_i^j,
	\end{align}
	where (a) follows from the Fubini's theorem 
	and (b) follows from \eqref{eq:Rij}. 
	Moreover,  the long-term average power is given by 
	$\mathbb{E}_g\left[\sum_{i=1}^{N}P_i\right]$ and we note the following: 
	\begin{align}\label{eq:MostGen-Simple}
	\mathbb{E}_g\left[\sum_{i=1}^{N}P_i\right]&=\mathbb{E}_g\left[\sum_{i=1}^{N}h_j^{-1}r^{-1}(R_i)\Pi_{l=1}^{i-1}\left(r^{-1}(R_l)+1\right)\right],  \nonumber\\
	&= \sum_{j=1}^{N}\mu_j\mathbb{E}_g\left[\sum_{i=1}^{N}h_j^{-1}r^{-1}(R_i)\Pi_{l=1}^{i-1}\left(r^{-1}(R_l)+1\right) \mid \mathcal{G}^j\right],\nonumber\\ 
	&\stackrel{(a)}{=} \sum_{j=1}^{N}\mu_j\sum_{i=1}^{N}\mathbb{E}_g\left[h_j^{-1}r^{-1}(R_i)\Pi_{l=1}^{i-1}\left(r^{-1}(R_l)+1\right) \mid \mathcal{G}^j\right],
	\nonumber\\
	&\stackrel{(b)}{\geq} \sum_{j=1}^{N}\mu_j\sum_{i=1}^{N}h_j^{-1}r^{-1}(\mathbb{E}_g\left[R_i\mid \mathcal{G}^j\right])\Pi_{l=1}^{i-1}\left(r^{-1}(\mathbb{E}_g\left[R_l\mid \mathcal{G}^j\right])+1\right),\nonumber\\
	&\stackrel{(c)}{=} \sum_{j=1}^{N}\mu_j\sum_{i=1}^{N}h_j^{-1}r^{-1}(R_i^j)\Pi_{l=1}^{i-1}\left(r^{-1}(R_l^j)+1\right)\stackrel{(d)}{=}\sum_{j=1}^{N}\mu_j\sum_{i=1}^{N}P_i^j, 
	\end{align}
	where (a) follows from the Fubini's theorem, (b) is due to the Jensen's inequality, (c) follows from \eqref{eq:Rij}, and (d) is because, we define  $P_i^j\triangleq h_j^{-1}r^{-1}(R_i^j)\Pi_{l=1}^{i-1}\left(r^{-1}(R_l^j)+1\right)$.

	Now, consider the following policy, which we refer to as the UB-NoCSIT policy:  {We design $N$ rate-tuples, each with $N$ elements, represented by $(R_1^j,\ldots,R_N^j)$, where $R_i^j$ is defined in \eqref{eq:Rij}.  In any given frame, we choose the $j^{\text{th}}$ rate-tuple, $(R_1^j,\ldots,R_N^j)$, with probability $\mu_j$, defined in \eqref{eq:muj}}. Note that the UB-NoCSIT policy is identical to the AI-SRP proposed in Section~\ref{eq:AISRP-NoCSIT}. Moreover,  in       the UB-NoCSIT policy, the long-term average throughput is given by $\sum_{j=1}^{N}\mu_j \sum_{i=1}^Nq_iR_i^j$, the long-term average power is given by $\sum_{j=1}^{N}\mu_j\sum_{i=1}^{N}P_i^j$ and the success probability is given by $\sum_{j=1}^{N}\mu_jq_j$, which is the same as the success probability in the most general policy described above. Now, from  \eqref{eq:MostGenRate} and \eqref{eq:MostGen-Simple}, it follows that the long-term average throughputs of both of the above policies are the same, to achieve which, the average power consumed in the UB-NoCSIT policy is less than or equal to that for the most general policy.  

From the above discussion, we conclude that for every fully general SRP, with a certain long-term average throughput,  there exists another policy, similar to the AI-SRP proposed in   Section~\ref{eq:AISRP-NoCSIT},   that consumes the same or a lower power with the same long-term average throughput as the fully general policy.  This implies that the optimal SRP that solves \eqref{eq:upper-bound-problemBC}  has the structure of of the proposed AI-SRP without CSIT in Section~\ref{eq:AISRP-NoCSIT}. 
	With this strategy, $\lim_{K\rightarrow \infty}\frac{1}{K}\sum_{k=1}^{K}\mathbb{E}[R_s(k)]$, $ \lim_{K\rightarrow \infty}\frac{1}{K}\sum_{k=1}^{K}\mathbb{E}[u(k)]$ and $\lim_{K\rightarrow \infty}\frac{1}{K}\sum_{k=1}^{K}\mathbb{E}[P(k)]$ in \eqref{eq:upper-bound-problemBC} are respectively given by \eqref{eq:AI-SRPobj} and the left-hand sides of \eqref{eq:avgAoIBC} and \eqref{eq:avgPBC}.
	Hence, the upper-bound problem  in \eqref{eq:upper-bound-problemBC} can be re-written as
	
	\begin{subequations}\label{eq:BC-Rate-AoI-UB}
		\begin{align}
		\underset{R_1^j,\ldots,R_N^j, \mu_j}{\text{max}} &\;\; \sum_{j=1}^{N}\sum_{i=1}^{N}\mu_jq_iR_i^j, \;\;&\\
		\text{subject to}&\;\eqref{eq:avgPBC},  \eqref{eq:AI-SRPBC-d} \;\sum_{j=1}^{N}q_j\mu_j\geq \frac{1}{2\alpha-1},\; \sum_{j=1}^{N}\mu_j=1,
		\end{align}
	\end{subequations}
The optimization problems \eqref{eq:BC-Rate-AoI-21} (or equivalently \eqref{eq:BC-Rate-AoI-22}) and \eqref{eq:BC-Rate-AoI-UB} are identical except  that the inequality $\sum_{j=1}^{N}q_j\mu_j\geq 1/\alpha$ in \eqref{eq:BC-Rate-AoI-21} is replaced by  $\sum_{j=1}^{N}q_j\mu_j\geq 1/(2\alpha-1)$ in \eqref{eq:BC-Rate-AoI-UB}. 
	Hence, we have,  $U_{\rm NoCSIT}=R_{\rm NoCSIT}^{\rm AI\_SRP}(2\alpha-1)$, where $R_{\rm NoCSIT}^{\rm AI\_SRP}(2\alpha-1)$ is given by \eqref{eq:BC-Rate-AoI-22}.

\bibliographystyle{ieeetran}
\bibliography{references}

\end{document}